\documentclass[onecolumn]{IEEEtran}

\usepackage{amsthm,amsmath,amssymb,amsfonts,bm}
\usepackage{tikz}
\usepackage{booktabs,multirow,nicematrix}
\usepackage[table,xcdraw]{xcolor}
\usepackage{adjustbox}
\usepackage[shortlabels]{enumitem}
\usepackage[square, comma, sort&compress, numbers]{natbib}
\usepackage{float,graphicx}
\usepackage[colorlinks,linkcolor=blue]{hyperref}
\usepackage{tabularx}
\usepackage{stmaryrd}      
\usepackage{etoolbox}      
\usepackage{mathtools}     
\usepackage{caption}       
\usepackage{pdflscape}     

\usepackage{cleveref}

\renewcommand{\eqref}[1]{\textcolor{blue}{(\ref{#1})}}

\newtheorem{theorem}{Theorem}[section]
\newtheorem{lemma}[theorem]{Lemma}
\newtheorem{proposition}[theorem]{Proposition}

\newtheorem{example}[theorem]{Example}
\newtheorem{remark}[theorem]{Remark}

\makeatletter
\@addtoreset{equation}{section}
\makeatother

\begin{document}
	\title{ Three classes of few-weight self-orthogonal projective codes and their applications for LCD codes}
	\author{Yue Huang, Zhonghao Liang, Chenlu Jia, Yongkang Wan and Qunying Liao
		\thanks{Corresponding author: Qunying Liao.  Emails:2907245673@qq.com; liangzhongh0807@163.com; chenlujiamath@163.com; 2475636261@qq.com;\\ qunyingliao@sicnu.edu.cn.}
		\thanks{Yue Huang, Zhonghao Liang, Chenlu Jia, Yongkang Wan and Qunying Liao are with College of Mathematical Sciences, Sichuan Normal University, Chengdu 610066, China.} 
		\thanks{This paper is supported by National Natural Science Foundation of China (12471494) and Natural Science Foundation of Sichuan Province (2024NSFSC2051).}
	}
	\maketitle
	
	\begin{abstract}
		Since self-orthogonal codes, few-weight codes and linear complementary dual codes(LCD codes, for short)  have nice applications in coding theory and cryptography,  they have received continuous attention. In 2024, by  introducing the notion of the augment code, Heng et al.\cite{28} constructed several classes of few-weight self-orthogonal codes basing on defining sets, which are introduced by Ding et al.\cite{8} in 2007.
		In this manuscript, for two classes of defining sets, we consider the corresponding augmented codes, 
		construct three classes of  four-weight self-orthogonal projective codes and determine  the parameters of their dual codes. And then, we construct three classes of LCD codes. In particular, we prove that there exists a class of these LCD codes  whose dual codes  are almost optimal  LCD codes according to the sphere packing bound.	
		

	\end{abstract} 

	\section{Introduction}

	
	Let $\mathbb{F}_{p^m}$ be the finite field with $p^m$ elements, where $p$ is a prime and $m$ is  a positive integer. An $[n,k,d]$ linear code $\mathcal{C}$  over $\mathbb{F}_{p^m}$ is a $k$-dimensional linear subspace of $\mathbb{F}_{p^m}^n$ with minimum (Hamming) distance $d.$ The weight $wt(\boldsymbol{c})$ of a codeword $\boldsymbol{c} \in \mathcal{C} $ is the number of nonzero coordinates in $\boldsymbol{c}.$ Let $A_i$  be the number of codewords $\boldsymbol{c} \in \mathcal{C}$ with 
	weight $i.$ Then $\left( 1,A_1,\cdots,A_n\right) $ is called the weight distribution of $\mathcal{C}$ and the polynomial $1+A_1z+\cdots+A_nz^n$ is called the weight enumerator of $\mathcal{C}.$ If $\#\left\lbrace i : A_i \neq 0, 1 \leq i \leq n \right\rbrace=t,$ then  $\mathcal{C}$ is   $t$-weight. If the Hamming weight of each codeword in $\mathcal{C}$ is divisible by $p$, then  $\mathcal{C}$  is  $p$-divisible.
	The dual code $\mathcal{C}^ \perp$ of $\mathcal{C}$ is defined by $\mathcal{C}^\perp=\left\lbrace\boldsymbol{x} \in \mathbb{F}_{p^m}^n: \boldsymbol{x} \cdot \boldsymbol{c}=0 ~\text{for all } \boldsymbol{c} \in \mathcal{C}\right\rbrace.$ If the dual code of $\mathcal{C}$ has the minimal distance $d ^ \perp \geq 3,$ then  $\mathcal{C}$ is a projective code. If $\mathcal{C} \subseteq \mathcal{C}^\perp,$ then  $\mathcal{C}$ is  self-orthogonal; and if  $\mathcal{C} \cap \mathcal{C}^{\perp}=\left\lbrace \textbf{0}\right\rbrace,$ then $\mathcal{C}$ is  a linear complementary dual code(LCD code, for short).

	It is well-known that any linear code over $\mathbb{F}_{p^m}$ has three fundamental parameters: the length $n,$ the  dimension $k$ and the minimum Hamming distance $d.$ A central goal in error-correcting code theory is to construct codes with good performance, i.e.,  simultaneously  large $n,k,d.$ Nevertheless, classical coding bounds show that these parameters cannot be maximized simultaneously, which introduces three separate definitions of optimality as follows: for any given  $\left[ n,k,d\right]_{p^m} $ linear code $\mathcal{C},$  if there is no  $\left[ n-1,k,d\right]_{p^m} $ linear code, then $\mathcal{C}$ is length-optimal; if there is no  $\left[ n,k+1,d\right]_{p^m} $ linear code, then $\mathcal{C}$ is dimension-optimal;  if there is no  $\left[ n,k,d+1\right]_{p^m} $ linear code, then $\mathcal{C}$ is distance-optimal. In 2025, Chen et al.\cite{49}, defined the distance-almost optimal $\left[ n,k,d\right]_{p^m} $ linear code. Similarly, the dimension-almost optimal   also can be defined, i.e., for any given $\left[ n,k,d\right]_{p^m} $ linear code $\mathcal{C}$, if there exists an $\left[ n,k+1,d\right]_{p^m} $ linear code, but there is no $\left[ n,k+2,d\right]_{p^m} $ code, then $\mathcal{C}$ is dimension-almost optimal. 
	So far, there are a lot of works on constructing optimal linear codes \cite{23,31,46,49,80,81,83}.
	
	Since self-orthogonal codes and few-weight codes play significant roles in coding theory and cryptography, they have been received considerable attention\cite{3,6,10,11,13,14,15,23,28,29,31,35,46,53,54,55}. Especially,
	in 2007, Ding et al.$\cite{8}$  introduced a new way to construct linear codes. Let $D=\left\lbrace d_1,d_2,\cdots,d_n \right\rbrace  \subseteq \mathbb{F}_{p^m}$ and defined $$\mathcal{C}_D=\left\lbrace (\mathrm{Tr}(bd_1),\mathrm{Tr}(bd_2),\cdots,\mathrm{Tr}(bd_n)) : b \in \mathbb{F}_{p^m}\right\rbrace, $$ where $\mathrm{Tr}$ is the trace function from $\mathbb{F}_{p^m}$ to $\mathbb{F}_p.$ And $D$ is called the defining set of  $\mathcal{C}_D.$ So far, there are many works on constructing  linear codes with good coding properties by choosing  the proper defining set $\cite{10,13,31,44,45,51,53,55}$.
	However, it is generally difficult to verify whether these codes obtained via such constructions are self-orthogonal. Recently, by establishing a sufficient condition such that  the $p$-divisible linear code $\mathcal{C}$ is self-orthogonal,  Li and Heng\cite{28} introduced the augmented code $\bar{\mathcal{C}}$ of $\mathcal{C},$ which is defined as
	\[
	\overline{\mathcal{C}}= \{ \bar{\textbf{c}} = \textbf{c} + \mu \mathbf{1} : \textbf{c} \in \mathcal{C}, \mu \in \mathbb{F}_{p^m} \},
	\]
	where $\mathbf{1}=(1,1,\cdots,1) \in \mathbb{F}_{p^m}^n.$  To date, there are only a few works devoted to investigating the properties of augmented codes constructed from defining sets\cite{58,14,15,28,39,53,60}. Moreover, existing results show that if the original defining-set code $\mathcal{C}_D$ is a few-weight code, its augmented code $\overline{\mathcal{C}_D}$ is not necessarily few-weight. Therefore, it's interesting to  find some proper defining set $D$ such that both $\mathcal{C}_D$ and $\overline{\mathcal{C}_D}$$
	$ are few-weight\cite{2025ZhangKaiZhu}.

	In 2024,  Du et al. \cite{35} introduced the linear code
	
	\begin{equation}
		\begin{split}
			\mathcal{C}_D=\left\lbrace \boldsymbol{c}(a , b)=(\mathrm{Tr}(ax+by))_{(x , y) \in D} :(a , b) \in \mathbb{F}_{p^m} \times \mathbb{F}_{p^m} \right\rbrace,
			\label{(1.1)} 
		\end{split}
	\end{equation}
	where 
	$$D=\left\lbrace (x , y) \in {\mathbb{F}^2_{p^m}} \setminus \left\lbrace (0 , 0)\right\rbrace : \mathrm{Tr} (\alpha xy +\beta y +x)=\gamma  \right\rbrace(\alpha \in \mathbb{F}_{p^m}^*,\beta \in \mathbb{F}_{p^m},\gamma \in \mathbb{F}_p ),$$
	and constructed a class of three-weight linear codes.
	
	In this manuscript, by taking $D=D_1$ or $D_2$ in \eqref{(1.1)}, where
	\begin{align}	
		D_1=&\left\lbrace (x , y) \in {\mathbb{F}_{p^m} \times \mathbb{F}_{p^m} }  : \mathrm{Tr} (\alpha xy +\beta y +x)=\gamma  \right\rbrace ,
		\label{(1.2)}\\
		D_2=&\left\lbrace (x , y) \in \mathbb{F}_{p^m}^* \times \mathbb{F}_{p^m}  : \mathrm{Tr} (\alpha xy  +x)=\gamma  \right\rbrace,
		\label{(1.3)}
	\end{align}
	we consider the augmented code $\overline{\mathcal{C}_{D_i}}(i=1,2),$ which is defined by
	\begin{equation}
		\begin{split}
			\overline{\mathcal{C}_{D_i}}=\left\lbrace \boldsymbol{c}(a,b,c)=\left( \mathrm{Tr}\left(ax+by \right)  \right)_{\left( x,y\right)\in D_i } +c\textbf{1}: a,b \in \mathbb{F}_{p^m} ,c \in \mathbb{F}_{p}\right\rbrace ,
			\label{(1.4)} 
		\end{split}
	\end{equation}
	construct three classes of  four-weight self-orthogonal projective codes and determine  the parameters of their dual codes. And then, we construct three classes of LCD codes. In particular, we prove that there exists a class of these LCD codes  whose dual codes  are almost optimal  LCD codes according to the sphere packing bound.	
	

	This manuscript is organized as follows. In Section \ref{section2},
	we introduce   some necessary  notations and  lemmas. In Section \ref{section3}, we give our main results and some auxiliary results. In Section \ref{section4}, we give the proofs of our main results and some examples. In Section \ref{section5}, we derive LCD codes  from the constructed self-orthogonal codes. In Section \ref{section6}, we conclude the whole manuscript.

	\section{Preliminaries}
	\label{section2}
	Throughout this manuscript, for the convenience, we fix some notations as follows.
	
	\begin{itemize}
		\item  $p$ is an odd prime and  $m \geq 2$ is a positive integer.
		
		\item $\mathbb{F}_{p^m}$ is the finite field with $p^m$ elements, where $\mathbb{F}_{p^m}^* = \mathbb{F}_{p^m} \setminus \{0\}$. 
		
		
		\item $\operatorname{Tr}(\cdot)$ is the trace map from $\mathbb{F}_{p^m}$ to $\mathbb{F}_p$.
		
		\item For any positive integer $t,$ $\boldsymbol{I}_t, \boldsymbol{0}_t$ and $\boldsymbol{1}_t$ denote the identity matrix of size $t \times t$, the zero matrix of size $t \times t$ and the all-1 vector of length $t,$ respectively. 
		
		
		\item For any positive integer $t,S_{p^t} = \{x^2 : x \in \mathbb{F}_{p^t}^*\}$ and  $NS_{p^t} = \mathbb{F}_{p^t}^* \setminus S_{p^t}$.

		
		
		
		\item $\chi$ and $\bar{\chi}$ denote the additive characters over $\mathbb{F}_{p^m}$ and  $\mathbb{F}_{p}$, respectively.
		
		\item $\eta$ and $\bar{\eta}$ denote the quadratic multiplicative characters over $\mathbb{F}_{p^m}$ and  $\mathbb{F}_{p}$, respectively.
		\item$\zeta_{p}$ be the primitive $p$-th root of unity and $\zeta_{p}=e^{\frac{2\pi i}{p}}(i=\sqrt{-1})$.
		\item $G$ and $\bar{G}$ denote the  Gauss sums over $\mathbb{F}_{p^m}$ and  $\mathbb{F}_{p}$, respectively.
		\item $t_1=\mathrm{Tr}(-\frac{\beta}{\alpha})-\gamma,t_2=\mathrm{Tr}(-\frac{b+a\beta}{\alpha})+c$ and $t_3=\mathrm{Tr}(-\frac{ab}{\alpha}).$
		\item $l_1=-\gamma,l_2=\mathrm{Tr}(-\frac{b}{\alpha})+c$ and $l_3=\mathrm{Tr}(-\frac{ab}{\alpha}).$
	\end{itemize}

	In this section, we firstly give some basic notations and known results on additive characters, quadratic characters, Gauss sums over $\mathbb{F}_{p^m}$, the Pless power moments, self-orthogonal codes and LCD codes.
	
	The following Lemmas \ref{lemma2.1}-\ref{lemma2.4} are crucial for computing the weight distributions in Section \ref{subsection4.1}. 
	\begin{lemma}{\rm(\cite{56}, Lemma 2.1)}
		\label{lemma2.1}
		For any $a,x$$ \in \mathbb{F}_{p^m}$, we have $$\sum_{x\in \mathbb{F}_{p^m}}\zeta_p^{\mathrm{Tr}(ax)}=  \begin{array}{c}
			\begin{cases}
				p^m , &\text{if} ~a=0\\
				0 , &\text{if} ~a\in \mathbb{F}_{p^m}^*
			\end{cases},
		\end{array}
		\sum_{x\in \mathbb{F}_{p^m}^*}\zeta_p^{\mathrm{Tr}(ax)}= \begin{array}{c}
			\begin{cases}
				p^m-1, &\text{if} ~a=0\\
				-1, &\text{if} ~a\in \mathbb{F}_{p^m}^*
			\end{cases}.
		\end{array} $$
	\end{lemma}

	\begin{lemma}{\rm(\cite{53})}
		\label{}
		For any $a\in \mathbb{F}_{p^m},$ we have $$\eta(a)=\begin{array}{c}
			\begin{cases}
				0 , &\text{if} ~a=0\\
				1 , &\text{if} ~a   \in S_{p^m}\\
				-1 , &\text{if} ~a  \in NS_{p^m} 
			\end{cases},
		\end{array} 
		\sum\limits_{a\in \mathbb{F}_{p^m}^*}\eta(a)=0.
		$$
	\end{lemma}

	\begin{lemma}{\rm(\cite{32}, Theorem 5.51)}
		\label{lemma2.3}
		For the Gauss sum  $\bar{G}(\bar\eta,\bar\chi)=\sum\limits_{c\in \mathbb{F}_{p}^*}\bar\eta(c)\bar\chi(c),$  we have
		$$\bar{G}^2(\bar\eta,\bar\chi)=\bar\eta (-1)p.$$
	\end{lemma}
	
	\begin{lemma}{\rm(\cite{32}, Theorem 5.33)}	\label{lemma2.4}
		Let $f(x)=a_2x^2+a_1x+a_0 \in \mathbb{F}_{p^m}\left[x \right] $ and $a_2 \in \mathbb{F}_{p^m}^* .$ Then 
		$$\sum_{x\in \mathbb{F}_{p^m}}\zeta_p^{\mathrm{Tr}(f(x))}=\zeta_p^{\mathrm{Tr}(a_0-a_1^2(4a_2)^{-1})}\eta(a_2)G(\eta,\chi).$$
	\end{lemma}
	
	The Pless power moments provided in the following Lemma \ref{lemma2.5} plays a key role in computing the  distance of the dual code.
	\begin{lemma} {\rm(\cite{18}, Page 259)}	\label{lemma2.5}
		Let $\mathcal{C}$  be an $\left[n,k \right] $ code over $\mathbb{F}_{p}$ with weight distribution $\left(1,A_1,\cdots , A_n \right) $ and $\mathcal{C}^\perp$ be the dual code with weight distribution $\left(1,A_1^\perp,\cdots , A_n^\perp \right) $, then the first four Pless power moments are as follows,
		
		\begin{align}
			&\sum\limits_{j=0}^{n}A_j=p^k,\label{(2.1)}\\
			&\sum\limits_{j=0}^{n}jA_j=p^{k-1}\left(pn-n-A_1^\perp \right),\label{(2.2)}\\
			&\sum\limits_{j=0}^{n}j^2A_j=p^{k-2}\left[\left(p-1 \right)n\left( pn-n+1\right)-\left(2pn-p-2n+2 \right)A_1^\perp +2A_2^\perp   \right] ,\label{(2.3)}\\
			&\sum_{j=0}^n j^3 A_j = p^{k-3} \left[ (p - 1)nE - FA_1^\perp + 6(pn - p - n + 2)A_2^\perp - 6A_3^\perp \right]\label{(2.4)},
		\end{align}
		where 
		$$E=p^2n^2 - 2pn^2 + 3pn - p + n^2 - 3n + 2$$
		and
		$$F=3p^2n^2 - 3p^2n - 6pn^2 + 12pn + p^2 - 6p + 3n^2 - 9n + 6.$$
	\end{lemma}

	The following Lemma \ref{lemma2.6} presents a sufficient condition for a  $p$-divisible code to be self-orthogonal.
	
	\begin{lemma} {\rm(\cite{28}, Theorem 1)}	\label{lemma2.6}
		Let $\mathcal{C}$ be an $\left[n,k,d \right] $ linear code over $\mathbb{F}_{p^m}.$ If $\boldsymbol{1_n} \in \mathcal{C}$ and $\mathcal{C}$ is $p$-divisible, then $\mathcal{C}$ is self-orthogonal.
	\end{lemma}

	The following Lemma \ref{lemma2.7} provides a method to construct LCD codes  based on self-orthogonal linear codes.
	\begin{lemma}{\rm(\cite{53}, Lemma 5.9)}	\label{lemma2.7}
		Let $\mathcal{C}$ be an $\left[ n ,k ,d\right] $ linear code generated by the matrix $\boldsymbol{G}$
		over $\mathbb{F}_{p^m}.$ If $\mathcal{C}$ is self-orthogonal, then the matrix $ \boldsymbol{\bar{G}} = [\boldsymbol{I}_k :\boldsymbol{G}]$
		generates an $\left[ n +k ,k \right] $ LCD code over $\mathbb{F}_{p^m}.$
	\end{lemma}

	The following Lemma \ref{lemma2.9} gives the well-known sphere packing bound for a linear code.
	\begin{lemma}{\rm(\cite{18}, Theorem 1.12.1)}	\label{lemma2.9}
		Let $\mathcal{C}$ be an $[n,k,d]_p$ linear code. Then $$p^n \geq p^k \sum_{i=0}^{\left\lfloor \frac{d-1}{2} \right\rfloor} \binom{n}{i} (p-1)^i, $$
		where \( \lfloor \cdot \rfloor \) denotes the floor function.
	\end{lemma}

	
	The following Lemma \ref{lemma2.12} will be used to calculate  $\# D_1$ in Section \ref{subsection4.1}.
	\begin{lemma}{\rm(\cite{35})}\label{lemma2.12}
		Let $\Omega_1=\sum\limits_{u\in \mathbb{F}_p^*}  \sum\limits_{x\in \mathbb{F}_{p^m}} \sum\limits_{y\in \mathbb{F}_{p^m}} \zeta_{p}^{\mathrm{Tr} (u\alpha xy +u\beta y +ux)-u\gamma},$ then 
		$$\Omega_1= \begin{array}{c}
			\begin{cases}
				p^m(p-1) , &\text{if} ~\gamma = \mathrm{Tr}(-\frac{\beta}{\alpha});\\
				-p^m , &\text{if} ~\gamma \neq \mathrm{Tr}(-\frac{\beta}{\alpha}).
			\end{cases}
		\end{array} $$
	\end{lemma}
	\section{ Three classes of four-weight self-orthogonal  projective codes}
	\label{section3}
	In this section, we give our main results, namely, three classes of  four-weight self-orthogonal projective linear codes,  and some auxiliary results will be used in Section \ref{subsection4.1}.
	
	\subsection{Main results}
	\label{subsection3.1}
	Let $\overline{\mathcal{C}_{D_1}}$ and $D_1$ be defined by \eqref{(1.4)} and \eqref{(1.2)} in Theorems \ref{theorem3.1}-\ref{theorem3.2}, respectively.
	\begin{theorem}\label{theorem3.1}
		If  $\gamma=\mathrm{Tr}(-\frac{\beta}{\alpha}),$ then $\overline{\mathcal{C}_{D_1}}$ is a $\left[p^{m-1}(p^m+p-1),2m+1,p^{2m-2}(p-1) \right]_p$ four-weight self-orthogonal linear code with the weight distribution given by Table \ref{Table1} , and the dual code ${\overline{\mathcal{C}_{D_1}}}^{\perp}$ is a $[p^{m-1}(p^m+p-1),$$p^{m-1}(p^m+p-1)-2m-1,3 ]_p $ code.  Furthermore ${\overline{\mathcal{C}_{D_1}}}$ is a projective linear code.
		
	\end{theorem}
	\begin{table}[H]
		\centering
		\caption{The weight distribution of \( \overline{\mathcal{C}_{D_1}} \)}
		\begin{tabular}{c c}
			\toprule
			weight & frequency \\
			\midrule
			0 & 1 \\
			\( p^{2m-2}(p-1) \) & \( p^{2m-1}+p^m-p^{m-1}-1 \) \\
			\( p^{m-1}(p^m-p^{m-1}+p-1) \) & \( (p^m-1)(2p^m-2p^{m-1}+p-1) \) \\
			\( p^{m-1}(p^m-p^{m-1}+p-2) \) & \( p^{m-1}(p^{m+2}-2p^{m+1}+p^m-p^2+2p-1) \) \\
			\( p^{m-1}(p^m+p-1) \) & \( p - 1 \) \\
			\bottomrule
		\end{tabular}
		\label{Table1}
	\end{table}
	
	\begin{theorem}\label{theorem3.2}
		If  $\gamma \neq \mathrm{Tr}(-\frac{\beta}{\alpha}),$ then $\overline{\mathcal{C}_{D_1}}$ is a $\left[p^{m-1}(p^m-1),2m+1,p^{m-1}(p^m-p^{m-1}-2) \right]_p$ four-weight self-orthogonal linear code with the the weight distribution given by Table \ref{TableA1}, and the dual code ${\overline{\mathcal{C}_{D_1}}}^{\perp}$ is a $[p^{m-1}(p^m-1),$$p^{m-1}(p^m-1)-2m-1,3 ]_p $ code. Furthermore ${\overline{\mathcal{C}_{D_1}}}$ is a projective linear code.
	\end{theorem}
	\begin{table}[H]
		\centering
		\caption{The weight distribution of \( \overline{\mathcal{C}_{D_1}} \)}
		\begin{tabular}{c c}
			\toprule
			weight & frequency \\
			\midrule
			0 & 1 \\
			\( p^{2m-2}(p-1) \) & \( \frac{(p^m-1)(p^{m+1}-p^m+2p^{m-1}+2)}{2} \) \\
			\( p^{m-1}(p^m-p^{m-1}-1) \) & \( (p^m-1)(2p^m-2p^{m-1}+p-1) \) \\
			\( p^{m-1}(p^m-p^{m-1}-2) \) & \( \frac{p^{m-1}(p^m-1)(p-1)(p-2)}{2} \) \\
			\( p^{m-1}(p^m-1) \) & \( p - 1 \) \\
			\bottomrule
		\end{tabular}
		\label{TableA1}
	\end{table}

	Let $\overline{\mathcal{C}_{D_2}}$ and $D_2$ be defined by \eqref{(1.4)} and \eqref{(1.3)} in Theorem \ref{theorem3.4}, respectively.
	\begin{theorem}\label{theorem3.4}
		If  $\gamma=0,$ then $\overline{\mathcal{C}_{D_2}}$ is a $\left[ p^{m-1}(p^m-1),2m+1,p^{m-1}(p^m-p^{m-1}-p+1)\right]_p$ four-weight self-orthogonal linear code with the weight distribution given by Table \ref{table2}, and the dual code ${\overline{\mathcal{C}_{D_2}}}^{\perp}$ is a $[p^{m-1}(p^m-1),$$p^{m-1}(p^m-1)-2m-1,3 ]_p$ code. Furthermore ${\overline{\mathcal{C}_{D_2}}}$ is a projective linear code.  
	\end{theorem}
	
	\begin{table}[H]
		\centering
		\caption{The weight distribution of \( \overline{\mathcal{C}_{D_2}} \)}
		\begin{tabular}{c c}
			\toprule
			weight & frequency \\
			\midrule
			0 & 1 \\
			\( p^{2m-2}(p-1) \) & \( (p^m-1)(2p^m-2p^{m-1}+1) \) \\
			\( p^{m-1}(p^m-p^{m-1}-1) \) & \( (p-1)(p^m-1)(p^m-p^{m-1}+1) \) \\
			\( p^{m-1}(p^m-p^{m-1}-p+1) \) & \( p^{m-1}(p^m-1) \) \\
			\( p^{m-1}(p^m-1) \) & \( p - 1 \) \\
			\bottomrule
			\label{table2}
		\end{tabular}
	\end{table}

	\begin{remark}\label{remark3.6}
		If $\gamma \neq 0,$ then $ \overline{\mathcal{C}_{D_2}}$  is a four-weight self-orthogonal projective linear code with the same parameters and weight distribution in  Theorem \ref{theorem3.2}. Therefore, we only give the proof of $\overline{\mathcal{C}_{D_2}}$ with $\gamma =0$.
	\end{remark}

	\subsection{Some auxiliary results}
	\label{subsection3.2}
	In this subsection, we give some  auxiliary results which will be used in the sequel. Let $n$ be the length of $\overline{\mathcal{C}_{D_i}}.$ Obviously, $\left( \mathrm{Tr}\left(ax+by \right)  \right)_{\left( x,y\right)\in D_i } +c\boldsymbol{1_n}$ is the zero codeword when $(a,b,c)=(0,0,0).$ Hence, we  default $(a,b,c) \neq (0,0,0)$ later.
	\begin{proposition}\label{prop3.7}
		Let $\Omega_2=\sum\limits_{v\in \mathbb{F}_p^*}  \sum\limits_{x\in \mathbb{F}_{p^m}} \sum\limits_{y\in \mathbb{F}_{p^m}} \zeta_{p}^{\mathrm{Tr} (vax+vby)+vc},$ then 
		$$\Omega_2= \begin{array}{c}
			\begin{cases}
				0 ,  &\text{if} ~a\neq 0,b \in \mathbb{F}_{p^m},c\in \mathbb{F}_p, \\
				&~~\text{or}~ a=0, b\neq 0,c \in \mathbb{F}_p;\\
				-p^{2m} , &\text{if} ~ a=b=0,c\neq0.
			\end{cases}
		\end{array} $$
	\end{proposition}

	\noindent\textbf{Proof.}
	~It's easy to get 
	
	$$	\Omega_2=\sum_{v\in \mathbb{F}_p^*}\zeta_{p}^{vc} \sum_{x\in \mathbb{F}_{p^m}}\zeta_{p}^{\mathrm{Tr} (vax)}\sum_{y\in \mathbb{F}_{p^m}}\zeta_{p}^{\mathrm{Tr} (vby)}.$$
	Next depending on the non-zeroness of $a,b$ or $c,$ we have the following three cases.

	$\bf{Case~ 1.}$~If $a \neq 0,b \in \mathbb{F}_{p^m}$ and $c \in \mathbb{F}_p$, then by Lemma \ref{lemma2.1}, we have  $\sum\limits_{x\in \mathbb{F}_{p^m}}\zeta_{p}^{\mathrm{Tr} (vax)}=0,$ thus,  $\Omega_2=0.$
	
	$\bf{Case~ 2.}$~If $a=0, b \neq 0 $ and $c \in \mathbb{F}_p$, then  by Lemma \ref{lemma2.1}, we have  $\sum\limits_{y\in \mathbb{F}_{p^m}}\zeta_{p}^{\mathrm{Tr} (vby)}=0,$ thus,  $\Omega_2=0.$

	$\bf{Case ~3.}$~If $a=b=0 $ and $c \neq 0$, then  by Lemma \ref{lemma2.1}, we have $\sum\limits_{v\in \mathbb{F}_p^*}\zeta_{p}^{vc}=-1,$ thus,  $\Omega_2=-p^{2m}.$
	
	From the above, we complete the proof of Proposition \ref{prop3.7}.\hfill$\square$

	\begin{proposition}\label{prop3.8}
		Let $\Omega_3=\sum\limits_{u\in \mathbb{F}_p^*}\sum\limits_{v\in \mathbb{F}_p^*}  \sum\limits_{x\in \mathbb{F}_{p^m}} \sum\limits_{y\in \mathbb{F}_{p^m}} \zeta_{p}^{\mathrm{Tr} (u\alpha xy +u\beta y +ux+vax+vby)-u \gamma +vc}, $ then 
		$$\Omega_3= \begin{array}{c}
			\begin{cases}
				p^m(p-1)^2, &\text{if}~ t_1=t_2=t_3=0;\\
				-p^m(p-1), &\text{if}~ 
				t_1=0,t_2 \neq 0, t_3=0,\\ &~~\text{or}~t_1\neq0,t_2=t_3=0, \\
				&~~\text{or}~  t_1=t_2=0,t_3\neq0,\\
				&~~\text{or}~t_1 \neq 0,t_3 \neq 0,t_2^2-4t_1t_3 \in NS_p;	\\
				p^m,&\text{if}~ t_1\neq0,t_2 \neq 0,t_3=0, \\
				&~~\text{or}~ t_1=0,t_2\neq 0,t_3 \neq 0,\\ &~~\text{or}~ t_1 \neq 0,t_3 \neq 0, 4t_1t_3-t_2^2=0;\\
				p^m(p+1), &\text{if}~  t_1\neq 0,t_3 \neq 0 ,t_2^2-4t_1t_3 \in S_p.
			\end{cases}
		\end{array}$$
	\end{proposition}
	\noindent\textbf{Proof.}
	By Lemma \ref{lemma2.1}, it's easy to get
	$$\begin{aligned}
		\Omega_3
		=&\sum_{u\in \mathbb{F}_p^*}\sum_{v\in \mathbb{F}_p^*}\zeta_{p}^{-u \gamma +vc}\sum_{x\in \mathbb{F}_{p^m}} \zeta_{p}^{\mathrm{Tr}((u+av)x)}\sum_{y\in \mathbb{F}_{p^m}}\zeta_{p}^{\mathrm{Tr}((u\alpha x+u \beta +vb)y)}\\
		=&\sum_{u\in \mathbb{F}_p^*}\sum_{v\in \mathbb{F}_p^*}\zeta_{p}^{-u \gamma +vc}\sum_{x\in \mathbb{F}_{p^m}\setminus \left\lbrace -\frac{u\beta+vb}{u\alpha} \right\rbrace  } \zeta_{p}^{\mathrm{Tr}((u+av)x)}\sum_{y\in \mathbb{F}_{p^m}}\zeta_{p}^{\mathrm{Tr}((u\alpha x+u \beta +vb)y)}\\
		&+p^m\sum_{u\in \mathbb{F}_p^*}\sum_{v\in \mathbb{F}_p^*}\zeta_{p}^{-u \gamma +vc}\zeta_{p}^{\mathrm{Tr}(-\frac{(u+av)(u\beta+vb)}{u\alpha})}\\
		=&p^m\sum_{u\in \mathbb{F}_p^*}\sum_{v\in \mathbb{F}_p^*}\zeta_{p}^{-u \gamma +vc}\zeta_{p}^{\mathrm{Tr}(-\frac{(u+av)(u\beta+vb)}{u\alpha})}\\
		=&p^m\sum_{u\in \mathbb{F}_p^*}\sum_{v\in \mathbb{F}_p^*}\zeta_{p}^{(\mathrm{Tr}(-\frac{\beta}{\alpha})-\gamma)u+(\mathrm{Tr}(-\frac{b+a\beta}{\alpha})+c)v+\mathrm{Tr}(-\frac{ab}{\alpha})\frac{v^2}{u}}\\
		=&p^m\sum_{u\in \mathbb{F}_p^*}\sum_{v\in \mathbb{F}_p^*}\zeta_{p}^{t_1u+t_2v+t_3\frac{v^2}{u}}.
	\end{aligned}
	$$
	Next depending on  $t_3=0$ or not, we have the following two cases.
	
	$\bf{Case ~1.}$ If $ t_3=0,$ then  $\Omega_3=p^m\sum\limits_{u\in \mathbb{F}_p^*}\sum\limits_{v\in \mathbb{F}_p^*}\zeta_{p}^{t_1u+t_2v}.$
	Now, depending on $t_i=0$ or not $(i=1,2),$  we  have

	$$\Omega_3= \begin{array}{c}
		\begin{cases}
			p^m(p-1)^2 ,& \text{if}~ t_1=t_2=0;\\
			-p^m(p-1) ,&\text{if} ~ t_1=0,t_2 \neq 0,\\
			&~~\text{or}~ t_1 \neq 0,t_2=0;\\
			p^m, &\text{if} ~ t_1 \neq 0,t_2 \neq 0.
		\end{cases}
	\end{array} $$

	$\bf{Case~ 2.}$ If $ t_3\neq 0,$ then by Lemma \ref{lemma2.4},  we have  $$\begin{aligned}
		\Omega_3=&p^m\sum_{u\in \mathbb{F}_p^*}\zeta_{p}^{t_1u}\sum_{v\in \mathbb{F}_p^*}\zeta_{p}^{t_2v+t_3\frac{v^2}{u}}\\
		=&p^m\sum_{u\in \mathbb{F}_p^*}\zeta_{p}^{t_1u}\left( \sum_{v\in \mathbb{F}_p}\zeta_{p}^{t_2v+t_3\frac{v^2}{u}}-1\right) \\
		=&p^m\sum_{u\in \mathbb{F}_p^*}\zeta_{p}^{t_1u}\sum_{v\in \mathbb{F}_p}\zeta_{p}^{t_2v+t_3\frac{v^2}{u}}-p^m\sum_{u\in \mathbb{F}_p^*}\zeta_{p}^{t_1u}\\
		=&p^m\sum_{u\in \mathbb{F}_p^*}\zeta_{p}^{t_1u}\zeta_{p}^{-\frac{t_2^2u}{4t_3}}\bar{\eta}(\frac{t_3}{u})\bar{G}(\bar{\eta},\bar{\chi})-p^m\sum_{u\in \mathbb{F}_p^*}\zeta_{p}^{t_1u}\\
		=&p^m\sum_{u\in \mathbb{F}_p^*}\zeta_{p}^{\frac{4t_1t_3-t_2^2}{4t_3}u}\bar{\eta}({t_3}{u})\bar{G}(\bar{\eta},\bar{\chi})-p^m\sum_{u\in \mathbb{F}_p^*}\zeta_{p}^{t_1u}.
	\end{aligned}$$
	Now depending on $t_1=0$ or not, we  have the following two cases.
	
	(1) If $t_1=0,$ then we have $\Omega_3=p^m\sum\limits_{u\in \mathbb{F}_p^*}\zeta_{p}^{\frac{-t_2^2}{4t_3}u}\bar{\eta}({t_3}{u})\bar{G}(\bar{\eta},\bar{\chi})-p^m(p-1).$
	If $t_2= 0,$ then  $\sum\limits_{u\in \mathbb{F}_p^*}\bar{\eta}({t_3}{u})=0,$ thus,
	$\Omega_3=-p^m(p-1).$
	If $t_2\neq 0,$ then
	$$\begin{aligned}
		\Omega_3=&p^m\sum_{u\in \mathbb{F}_p^*}\zeta_{p}^{\frac{-t_2^2}{4t_3}u}\bar{\eta}({t_3}{u})\bar{G}(\bar{\eta},\bar{\chi})-p^m(p-1)\\
		=&p^m\sum_{u\in \mathbb{F}_p^*}\zeta_{p}^{\frac{-t_2^2}{4t_3}u}\bar{\eta}(\frac{-t_2^2}{4t_3}u)\bar{\eta}(-\frac{4t_3^2}{t_2^2})\bar{G}(\bar{\eta},\bar{\chi})-p^m(p-1)\\
		=& p^m\bar{G}^2(\bar{\eta},\bar{\chi}) \bar{\eta}(-1)-p^m(p-1)\\
		=&p^{m+1} \bar{\eta}^2(-1)-p^m(p-1)\\
		=&p^{m+1}-p^m(p-1)\\
		=&p^m.
	\end{aligned}
	$$
	
	(2) If $t_1\neq 0,$ then $\Omega_3= p^m\sum\limits_{u\in \mathbb{F}_p^*}\zeta_{p}^{\frac{4t_1t_3-t_2^2}{4t_3}u}\bar{\eta}({t_3}{u})\bar{G}(\bar{\eta},\bar{\chi})+p^m.$
	If $4t_1t_3-t_2^2=0,$ then  $\sum\limits_{u\in \mathbb{F}_p^*}\bar{\eta}(t_3u)=0,$ thus, $\Omega_3=p^m.$
	If $t_2^2-4t_1t_3 \in S_p,$ then
	$$\begin{aligned}
		\Omega_3=&p^m\sum_{u\in \mathbb{F}_p^*}\zeta_{p}^{\frac{4t_1t_3-t_2^2}{4t_3}u}\bar{\eta}(\frac{4t_1t_3-t_2^2}{4t_3}u)\bar{\eta}(\frac{4t_3^2}{4t_1t_3-t_2^2})\bar{G}(\bar{\eta},\bar{\chi})+p^m\\
		=&p^m\bar{G}^2\bar{\eta}(-1)+p^m\\
		=&p^m(p+1).
	\end{aligned}$$
	If $t_2^2-4t_1t_3 \in NS_p,$ then $$
	\begin{aligned}	
		\Omega_3=&p^m\sum_{u\in \mathbb{F}_p^*}\zeta_{p}^{\frac{4t_1t_3-t_2^2}{4t_3}u}\bar{\eta}(\frac{4t_1t_3-t_2^2}{4t_3}u)\bar{\eta}(\frac{4t_3^2}{4t_1t_3-t_2^2})\bar{G}(\bar{\eta},\bar{\chi})+p^m\\
		=&-p^m\bar{G}^2\bar{\eta}(-1)+p^m\\
		=&p^m(1-p).
	\end{aligned}$$
	
	From the above, we complete the proof of Proposition \ref{prop3.8}.\hfill$\square$
	
	
	\begin{proposition}\label{prop3.9}
		Let $\Omega_4=\sum\limits_{v\in \mathbb{F}_p^*}  \sum\limits_{x\in \mathbb{F}_{p^m}^*} \sum\limits_{y\in \mathbb{F}_{p^m}} \zeta_{p}^{\mathrm{Tr} (vax+vby)+vc},$ then 
		$$\Omega_4= \begin{array}{c}
			\begin{cases}
				-p^m(p-1) ,& \text{if} ~a\neq 0,b=c=0;\\
				p^m ,&\text{if} ~ a\neq 0, b=0,c\neq0;\\
				0, &\text{if} ~a\in \mathbb{F}_{p^m}, b\neq 0, c \in \mathbb{F}_{p};\\
				-p^m(p^m-1),&\text{if} ~a=b=0,c\neq 0.
			\end{cases}
		\end{array} $$
	\end{proposition}
	\noindent\textbf{Proof.}
	It's easy to get
	$$
	\Omega_4=\sum_{v\in \mathbb{F}_p^*}\zeta_{p}^{vc} \sum_{x\in \mathbb{F}_{p^m}^*}\zeta_{p}^{\mathrm{Tr} (vax)}\sum_{y\in \mathbb{F}_{p^m}}\zeta_{p}^{\mathrm{Tr} (vby)}.
	$$	
	Next depending on the non-zeroness of $a,b$ or $c,$ we have the following four cases.
	
	$\bf{Case~ 1.}$~If $b \neq 0,a \in \mathbb{F}_{p^m}$ and $c \in \mathbb{F}_p$, then by Lemma \ref{lemma2.1}, we have $\sum\limits_{y\in \mathbb{F}_{p^m}}\zeta_{p}^{\mathrm{Tr} (vby)}=0, $ thus, $\Omega_4=0.$

	$\bf{Case~ 2.}$~If $b=0, a \neq 0 $ and $c =0$, then  by Lemma \ref{lemma2.1}, we have $\sum\limits_{x\in \mathbb{F}_{p^m}^*}\zeta_{p}^{\mathrm{Tr} (vax)}=-1, $ 
	thus, $ \Omega_4=-p^m(p-1).$

	$\bf{Case~ 3.}$~If $b=0, a \neq 0$ and $c \neq 0$, then by Lemmas \ref{lemma2.1}, we have $\sum\limits_{x\in \mathbb{F}_{p^m}^*}\zeta_{p}^{\mathrm{Tr} (vax)}=-1 $ and $\sum\limits_{v\in \mathbb{F}_p^*}\zeta_{p}^{vc}=-1 $
	, thus, $ \Omega_4=p^m.$

	$\bf{Case~ 4.}$~If $b=a=0 $ and $c \neq 0$, then by Lemma \ref{lemma2.1}, we have $ \sum\limits_{v\in \mathbb{F}_p^*}\zeta_{p}^{vc}=-1,$ thus, $\Omega_4=-p^{m}(p^m-1). $
	
	From the above, we complete the proof of Proposition \ref{prop3.9}.\hfill$\square$
	\begin{proposition}\label{prop3.10}
		Let $\Omega_5=\sum\limits_{u\in \mathbb{F}_p^*}\sum\limits_{v\in \mathbb{F}_p^*}  \sum\limits_{x\in \mathbb{F}_{p^m}^*} \sum\limits_{y\in \mathbb{F}_{p^m}} \zeta_{p}^{\mathrm{Tr} (u\alpha xy  +ux+vax+vby)-u \gamma +vc}, $ then 
		$$\Omega_5= \begin{array}{c}
			\begin{cases}
				0,&\text{if}~b=0;\\
				p^m(p-1)^2, &\text{if}~ l_1=l_2=l_3=0,b\neq0;\\
				-p^m(p-1), &\text{if}~ 
				l_1=0,l_2 \neq 0, l_3=0,b\neq0,\\ &~~\text{or}~l_1\neq0,l_2=l_3=0,b\neq0 ,\\
				&~~\text{or}~  l_1=l_2=0,l_3\neq0,b\neq0,\\
				&~~\text{or}~l_1 \neq 0,l_3 \neq 0,b\neq0,l_2^2-4l_1l_3 \in NS_p;\\	
				p^m,&\text{if}~ l_1\neq0,l_2 \neq 0,l_3=0,b\neq0, \\
				&~~\text{or}~ l_1=0,l_2\neq 0,l_3 \neq 0,b\neq0 ,\\
				&~~\text{or}~ l_1 \neq 0,
				l_3 \neq 0,b\neq0,4l_1l_3-l_2^2=0;\\
				p^m(p+1), &\text{if}~  l_1\neq 0, l_3 \neq 0,b\neq0,l_2^2-4l_1l_3 \in S_p.
			\end{cases}
		\end{array} $$
	\end{proposition}

	\noindent\textbf{Proof.}
	It's easy to get $$\Omega_5=\sum_{u\in \mathbb{F}_p^*}\sum_{v\in \mathbb{F}_p^*}\zeta_{p}^{-u \gamma +vc}\sum_{x\in \mathbb{F}_{p^m}^*} \zeta_{p}^{\mathrm{Tr}((u+av)x)}\sum_{y\in \mathbb{F}_{p^m}}\zeta_{p}^{\mathrm{Tr}((u\alpha x+vb)y)}. $$
	If $b=0,$  then by Lemma \ref{lemma2.1}, we have $\sum\limits_{y\in \mathbb{F}_{p^m}}\zeta_{p}^{\mathrm{Tr}(u\alpha x y)}=0, $ thus, $\Omega_5=0. $
	If $b\neq0,$ then $ -\frac{vb}{u \alpha} \neq 0.$ Furthermore, by Lemma \ref{lemma2.1}, we have 
	$$\begin{aligned}
		\Omega_5
		=&\sum_{u\in \mathbb{F}_p^*}\sum_{v\in \mathbb{F}_p^*}\zeta_{p}^{-u \gamma +vc}\sum_{x\in \mathbb{F}_{p^m}^*} \zeta_{p}^{\mathrm{Tr}((u+av)x)}\sum_{y\in \mathbb{F}_{p^m}}\zeta_{p}^{\mathrm{Tr}((u\alpha x+vb)y)}\\
		=&\sum_{u\in \mathbb{F}_p^*}\sum_{v\in \mathbb{F}_p^*}\zeta_{p}^{-u \gamma +vc}\sum_{x\in \mathbb{F}_{p^m}\setminus \left\lbrace 0,-\frac{vb}{u\alpha} \right\rbrace  } \zeta_{p}^{\mathrm{Tr}((u+av)x)}\sum_{y\in \mathbb{F}_{p^m}}\zeta_{p}^{\mathrm{Tr}((u\alpha x +vb)y)}\\
		&+p^m\sum_{u\in \mathbb{F}_p^*}\sum_{v\in \mathbb{F}_p^*}\zeta_{p}^{-u \gamma +vc}\zeta_{p}^{\mathrm{Tr}(-\frac{(u+av)vb}{u\alpha})}\\
		=&p^m\sum_{u\in \mathbb{F}_p^*}\sum_{v\in \mathbb{F}_p^*}\zeta_{p}^{-u \gamma +vc}\zeta_{p}^{\mathrm{Tr}(-\frac{(u+av)vb}{u\alpha})}\\
		=&p^m\sum_{u\in \mathbb{F}_p^*}\sum_{v\in \mathbb{F}_p^*}\zeta_{p}^{-u\gamma + (\mathrm{Tr}(-\frac{b}{\alpha})+c)v+\mathrm{Tr}(-\frac{ab}{\alpha})\frac{v^2}{u}}\\
		=&p^m\sum_{u\in \mathbb{F}_p^*}\sum_{v\in \mathbb{F}_p^*}\zeta_{p}^{l_1u + l_2v+\frac{l_3}{u}v^2}.
	\end{aligned}
	$$
	In the similar proof as that of Proposition \ref{prop3.8}, we can obtain 
	$$\Omega_5= \begin{array}{c}
		\begin{cases}
			p^m(p-1)^2, &\text{if}~ l_1=l_2=l_3=0,b\neq0;\\
			-p^m(p-1), &\text{if}~ 
			l_1=0,l_2 \neq 0, l_3=0,b\neq0,\\ &~~\text{or}~l_1\neq0,l_2=l_3=0,b\neq0 ,\\
			&~~\text{or}~  l_1=l_2=0,l_3\neq0,b\neq0,\\
			&~~\text{or}~l_1 \neq 0,l_3 \neq 0,b\neq0,l_2^2-4l_1l_3 \in NS_p;\\	
			p^m,&\text{if}~ l_1\neq0,l_2 \neq 0,l_3=0,b\neq0, \\
			&~~\text{or}~ l_1=0,l_2\neq 0,l_3 \neq 0,b\neq0 ,\\
			&~~\text{or}~ l_1 \neq 0,
			l_3 \neq 0,b\neq0,4l_1l_3-l_2^2=0;\\
			p^m(p+1), &\text{if}~  l_1\neq 0, l_3 \neq 0,b\neq0,l_2^2-4l_1l_3 \in S_p.
		\end{cases}
	\end{array} $$
	
	From the above, we complete the proof of Proposition \ref{prop3.10}.\hfill$\square$

	\section{The proofs of main results }
	\label{section4}
	In this section, we give the proofs of Theorems \ref{theorem3.1}-\ref{theorem3.4} and some corresponding examples.
	\subsection{The proofs of Theorems \ref{theorem3.1}-\ref{theorem3.4}}
	\label{subsection4.1}
	In this subsection, we give the proofs of  Theorems \ref{theorem3.1}-\ref{theorem3.4} basing on some auxiliary results in Section \ref{subsection3.2}.
	
	\noindent\textbf{The proofs of Theorems \ref{theorem3.1}-\ref{theorem3.2}.}
	
	According to the definition of $\overline{\mathcal{C}_{D_1}},$ we know that the length of $\overline{\mathcal{C}_{D_1}}$ equals to 
	$$\begin{aligned}
		n_1=&\# D_1 \\
		=&\frac{1}{p}\sum_{u\in \mathbb{F}_p}  \sum_{x\in \mathbb{F}_{p^m}} \sum_{y\in \mathbb{F}_{p^m}} \zeta_{p}^{u({\mathrm{Tr} (\alpha xy +\beta y +x)-\gamma})}\\
		=&\frac{1}{p}\sum_{u\in \mathbb{F}_p^*}  \sum_{x\in \mathbb{F}_{p^m}} \sum_{y\in \mathbb{F}_{p^m}} \zeta_{p}^{\mathrm{Tr} (u\alpha xy +u\beta y +ux)-u\gamma}+p^{2m-1}\\
		=&\frac{1}{p}\Omega_1+p^{2m-1}.
	\end{aligned}
	$$
	Now by Lemma \ref{lemma2.12}, we have 
	$$n_1= \begin{array}{c}
		\begin{cases}
			p^{2m-1}+p^{m}-p^{m-1}, &\text{if} ~\gamma = \mathrm{Tr}(-\frac{\beta}{\alpha});\\
			p^{2m-1}-p^{m-1} , &\text{if} ~\gamma \neq \mathrm{Tr}(-\frac{\beta}{\alpha}).
		\end{cases}
	\end{array} $$
	
	Let $N_1=\#\left\lbrace (x , y) \in {\mathbb{F}^2_{p^m}}  : \mathrm{Tr} (\alpha xy +\beta y +x)=\gamma  ~\text{and}~ \mathrm{Tr}(ax+by)+c=0\right\rbrace,$
	then the weight of the nonzero codeword \textbf{c}    in $\overline{\mathcal{C}_{D_1}}$ is 
	$$\begin{aligned}
		wt_1(\textbf{c})=&n_1-N_1\\
		=&n_1-\frac{1}{p^2}\sum_{u\in \mathbb{F}_p}\sum_{v\in \mathbb{F}_p}  \sum_{x\in \mathbb{F}_{p^m}} \sum_{y\in \mathbb{F}_{p^m}} \zeta_{p}^{u({\mathrm{Tr} (\alpha xy +\beta y +x)-\gamma})}\zeta_{p}^{v(\mathrm{Tr}(ax+by)+c)}\\
		=&n_1-\frac{1}{p^2}(p^{2m}+\sum_{u\in \mathbb{F}_p^*}  \sum_{x\in \mathbb{F}_{p^m}} \sum_{y\in \mathbb{F}_{p^m}} \zeta_{p}^{u({\mathrm{Tr} (\alpha xy +\beta y +x)-\gamma})}+\sum_{v\in \mathbb{F}_p^*}\sum_{x\in \mathbb{F}_{p^m}} \sum_{y\in \mathbb{F}_{p^m}} \zeta_{p}^{v(\mathrm{Tr}(ax+by)+c)}\\
		&+\sum_{u\in \mathbb{F}_p^*}\sum_{v\in \mathbb{F}_p^*}  \sum_{x\in \mathbb{F}_{p^m}} \sum_{y\in \mathbb{F}_{p^m}} \zeta_{p}^{u({\mathrm{Tr} (\alpha xy +\beta y +x)-\gamma})}\zeta_{p}^{v(\mathrm{Tr}(ax+by)+c)})\\
		=&n_1-\frac{1}{p^2}(p^{2m}+\Omega_1+\Omega_2+\Omega_3).
	\end{aligned}
	$$
	
	Now depending on $\gamma=\mathrm{Tr}(-\frac{\beta}{\alpha})$ or not, by Lemma \ref{lemma2.12}, Propositions \ref{prop3.7}-\ref{prop3.8} and  directly computing, we have the following two cases.
	
	$\bf{Case~ 1.}$ If $\gamma=\mathrm{Tr}(-\frac{\beta}{\alpha})$, then
	$$\begin{aligned}
		\Omega_1+\Omega_2+\Omega_3
		=\begin{array}{c}
			\begin{cases}
				p^{m+2}-p^{m+1},  &\text{if}~a\neq 0,b \in \mathbb{F}_{p^m},c\in \mathbb{F}_p , t_2=t_3=0,\\
				&~~\text{or}~a=0, b\neq 0,c \in \mathbb{F}_p,t_2=t_3=0;\\
				0, &\text{if}~a\neq 0,b \in \mathbb{F}_{p^m},c\in \mathbb{F}_p , t_2\neq 0,t_3=0,\\
				&~~\text{or}~a=0, b\neq 0,c \in \mathbb{F}_p,t_2\neq0,t_3=0,\\
				&~~\text{or}~a\neq 0,b \in \mathbb{F}_{p^m},c\in \mathbb{F}_p ,t_2=0,t_3\neq0;\\
				p^{m+1},&\text{if}~a\neq 0,b \in \mathbb{F}_{p^m},c\in \mathbb{F}_p ,t_2\neq0,t_3\neq0;\\
				-p^{2m},&\text{if}~ a=b=0,c\neq0,t_2 \neq 0,t_3=0.
			\end{cases}
		\end{array} \\
	\end{aligned}
	$$
	
	Furthermore, $$wt_1(\textbf{c})=\begin{array}{c}
		\begin{cases}
			p^{2m-2}(p-1),&\text{if}~a\neq 0,b \in \mathbb{F}_{p^m},c\in \mathbb{F}_p , t_2=t_3=0,\\
			&~~\text{or}~a=0, b\neq 0,c \in \mathbb{F}_p,t_2=t_3=0;\\
			p^{m-1}(p^m-p^{m-1}+p-1), &\text{if}~a\neq 0,b \in \mathbb{F}_{p^m},c\in \mathbb{F}_p , t_2\neq 0,t_3=0,\\
			&~~\text{or}~a=0, b\neq 0,c \in \mathbb{F}_p,t_2\neq0,t_3=0,\\
			&~~\text{or}~a\neq 0,b \in \mathbb{F}_{p^m},c\in \mathbb{F}_p ,t_2=0,t_3\neq0;\\
			p^{m-1}(p^m-p^{m-1}+p-2),&\text{if}~a\neq 0,b \in \mathbb{F}_{p^m},c\in \mathbb{F}_p ,t_2 \neq0,t_3\neq0;\\
			p^{m-1}(p^m+p-1),&\text{if}~ a=b=0,c\neq0,t_2\neq0,t_3=0.
		\end{cases}
	\end{array} \\
	$$
	
	$\bf{Case ~2.}$ If $\gamma\neq\mathrm{Tr}(-\frac{\beta}{\alpha})$, then
	$$\begin{aligned}
		\Omega_1+\Omega_2+\Omega_3=
		\begin{array}{c}
			\begin{cases}
				-p^{m+1}, &\text{if}~a\neq 0,b \in \mathbb{F}_{p^m},c\in \mathbb{F}_p , t_2=t_3=0,\\
				&~~\text{or}~a\neq 0,b \in \mathbb{F}_{p^m},c\in \mathbb{F}_p , t_3\neq0,t_2^2-4t_1t_3 \in NS_p,\\
				&~~\text{or}~a=0, b\neq 0,c \in \mathbb{F}_p,t_2=t_3=0;\\
				0, &\text{if}~a\neq 0,b \in \mathbb{F}_{p^m},c\in \mathbb{F}_p , t_2\neq 0,t_3=0,\\
				&~~\text{or}~a\neq 0,b \in \mathbb{F}_{p^m},c\in \mathbb{F}_p , t_3\neq0,4t_1t_3 -t_2^2 =0,\\
				&~~\text{or}~a= 0,b \neq 0,c\in \mathbb{F}_p , t_2\neq 0,t_3=0;\\
				p^{m+1},&\text{if}~a\neq 0,b \in \mathbb{F}_{p^m},c\in \mathbb{F}_p ,t_3\neq0,t_2^2-4t_1t_3 \in S_p;\\
				-p^{2m},&\text{if}~ a=b=0,c\neq0,t_2\neq0,t_3=0.
			\end{cases}
		\end{array} \\
	\end{aligned}
	$$
	
	Furthermore, 
	
	$$
	{wt}_1(\mathbf{c})=
	\begin{cases}
		p^{2m-2}(p-1), \quad &\text{if } a\neq 0,\ b \in \mathbb{F}_{p^m},\ c\in \mathbb{F}_p,\ t_2=t_3=0,\\
		&~~\text{or } a\neq 0,\ b \in \mathbb{F}_{p^m},\ c\in \mathbb{F}_p,\ t_3\neq0,\ t_2^2-4t_1t_3 \in NS_p,\\
		&~~\text{or } a=0,\ b\neq 0,\ c \in \mathbb{F}_p,\ t_2=t_3=0;\\
		p^{m-1}(p^m-p^{m-1}-1), \quad &\text{if } a\neq 0,\ b \in \mathbb{F}_{p^m},\ c\in \mathbb{F}_p,\ t_2\neq 0,\ t_3=0,\\
		&~~\text{or } a\neq 0,\ b \in \mathbb{F}_{p^m},\ c\in \mathbb{F}_p,\ t_3\neq0,\ 4t_1t_3-t_2^2 =0,\\
		&~~\text{or } a= 0,\ b \neq 0,\ c\in \mathbb{F}_p,\ t_2\neq 0,\ t_3=0;\\
		p^{m-1}(p^m-p^{m-1}-2), \quad &\text{if } a\neq 0,\ b \in \mathbb{F}_{p^m},\ c\in \mathbb{F}_p,\ t_3\neq0,\ t_2^2-4t_1t_3 \in S_p;\\
		p^{m-1}(p^m-1), \quad &\text{if } a=b=0,\ c\neq0,\ t_2\neq0,\ t_3=0.
	\end{cases}
	$$
	
	Now if  $\gamma=\mathrm{Tr}(-\frac{\beta}{\alpha}),$ then the weight $w_1=p^{2m-2}(p-1),w_2=p^{m-1}(p^m-p^{m-1}+p-1),w_3=p^{m-1}(p^m-p^{m-1}+p-2)$ and $w_4=p^{m-1}(p^m+p-1),$ and so we can get
	$$\begin{aligned}
		&A_{w_1}=p^{2m-1}+p^m-p^{m-1}-1, A_{w_2}=(p^m-1)(2p^m-2p^{m-1}+p-1),\\
		&A_{w_3}=p^{m-1}(p^{m+2}-2p^{m+1}+p^m-p^2+2p-1),A_{w_4}=p-1.
	\end{aligned}$$
	If	 $\gamma \neq \mathrm{Tr}(-\frac{\beta}{\alpha}),$ then the weight $\tilde{w}_1=p^{2m-2}(p-1),\tilde{w}_2=p^{m-1}(p^m-p^{m-1}-1),\tilde{w}_3=p^{m-1}(p^m-p^{m-1}-2)$ and $\tilde{w}_4=p^{m-1}(p^m-1),$ and so we can get $$\begin{aligned}
			&A_{\tilde{w}_1}= \frac{(p^m-1)(p^{m+1}-p^m+2p^{m-1}+2)}{2} , A_{\tilde{w}_2}=(p-1)(p^m-1)(2p^{m-1}+1),\\
			&A_{\tilde{w}_3}=\frac{p^{m-1}(p^m-1)(p-1)(p-2)}{2},A_{\tilde{w}_4}=p-1.
		\end{aligned}$$

	Note that $2 \leq m,$ it's easy to get $p \mid w_i$ and  $p \mid \tilde{w}_i (i=1,2,3,4).$ Hence, by Lemma \ref{lemma2.6}, the corresponding code  $\overline{\mathcal{C}_{D_1}}$ is  self-orthogonal whether $\gamma=\mathrm{Tr}(-\frac{\beta}{\alpha})$ or not.
	
	Finally,  we determine the distance of the dual code ${\overline{\mathcal{C}_{D_1}}}^{\perp}$ for  $\gamma = \mathrm{Tr}(-\frac{\beta}{\alpha})  $   and $\gamma \neq \mathrm{Tr}(-\frac{\beta}{\alpha})$. By Lemma \ref{lemma2.5} \eqref{(2.1)}-\eqref{(2.4)} and directly computing, when $\gamma = \mathrm{Tr}(-\frac{\beta}{\alpha})$, we have $A_1^{\perp}=A_2^{\perp}=0$ and $A_3^{\perp}>0$, which means that  $d({\overline{\mathcal{C}_{D_1}}}^{\perp})=3.$ When $\gamma \neq \mathrm{Tr}(-\frac{\beta}{\alpha})$, we have $A_1^{\perp}=A_2^{\perp}=0$ and $A_3^{\perp}>0$, which means that $d({\overline{\mathcal{C}_{D_1}}}^{\perp})=3.$ And by the definition of the projective code,  ${\overline{\mathcal{C}_{D_1}}}$ is a projective linear code whether $\gamma=\mathrm{Tr}(-\frac{\beta}{\alpha})$ or not. 
	
	From the above, the proofs of Theorems \ref{theorem3.1}-\ref{theorem3.2} are complete. \hfill$\square$

	\noindent\textbf{The proof of Theorem \ref{theorem3.4}}.
	
	According to the definition of $\overline{\mathcal{C}_{D_2}},$ the length of $\overline{\mathcal{C}_{D_2}}$ equals to
	$$\begin{aligned}
		n_2=&\# D_2 \\
		=&\frac{1}{p}\sum_{u\in \mathbb{F}_p}  \sum_{x\in \mathbb{F}_{p^m}^*} \sum_{y\in \mathbb{F}_{p^m}} \zeta_{p}^{u({\mathrm{Tr} (\alpha xy +x)-\gamma})}\\
		=&\frac{1}{p}\sum_{u\in \mathbb{F}_p^*}  \sum_{x\in \mathbb{F}_{p^m}^*} \sum_{y\in \mathbb{F}_{p^m}} \zeta_{p}^{{\mathrm{Tr} (u\alpha xy +ux)-u\gamma}}+p^{m-1}(p^m-1)\\
		=&\frac{1}{p}\sum_{u\in \mathbb{F}_p^*}\sum_{x\in \mathbb{F}_{p^m}^*} \zeta_{p}^{{\mathrm{Tr} (ux)-u\gamma}}  \sum_{y\in \mathbb{F}_{p^m}} \zeta_{p}^{{\mathrm{Tr} (u\alpha xy )}}+p^{m-1}(p^m-1)\\
		=&p^{m-1}(p^m-1).
	\end{aligned}
	$$
	Let $N_2=\#\left\lbrace (x , y) \in {\mathbb{F}_{p^m}^*} \times {\mathbb{F}_{p^m}}  : \mathrm{Tr} (\alpha xy +x)=\gamma  ~\text{and}~ \mathrm{Tr}(ax+by)+c=0\right\rbrace,$
	then the weight of the nonzero codeword \textbf{c} in $\overline{\mathcal{C}_{D_2}}$ is 
	$$\begin{aligned}
		wt_2(\textbf{c})=&n_2-N_2\\
		=&n_2-\frac{1}{p^2}\sum_{u\in \mathbb{F}_p}\sum_{v\in \mathbb{F}_p}  \sum_{x\in \mathbb{F}_{p^m}^*} \sum_{y\in \mathbb{F}_{p^m}} \zeta_{p}^{u({\mathrm{Tr} (\alpha xy  +x)-\gamma})}\zeta_{p}^{v(\mathrm{Tr}(ax+by)+c)}\\
		=&n_2-\frac{1}{p^2}(p^{m}(p^m-1)+\sum_{u\in \mathbb{F}_p^*}  \sum_{x\in \mathbb{F}_{p^m}^*} \sum_{y\in \mathbb{F}_{p^m}} \zeta_{p}^{u({\mathrm{Tr} (\alpha xy  +x)-\gamma})}+\sum_{v\in \mathbb{F}_p^*}\sum_{x\in \mathbb{F}_{p^m}^*} \sum_{y\in \mathbb{F}_{p^m}} \zeta_{p}^{v(\mathrm{Tr}(ax+by)+c)}\\
		&+\sum_{u\in \mathbb{F}_p^*}\sum_{v\in \mathbb{F}_p^*}  \sum_{x\in \mathbb{F}_{p^m}^*} \sum_{y\in \mathbb{F}_{p^m}} \zeta_{p}^{u({\mathrm{Tr} (\alpha xy  +x)-\gamma})}\zeta_{p}^{v(\mathrm{Tr}(ax+by)+c)})\\
		=&n_2-\frac{1}{p^2}(p^{m}(p^m-1)+\Omega_4+\Omega_5).
	\end{aligned}
	$$
	
	Now depending on $\gamma=0$ or not, by Propositions \ref{prop3.9}-\ref{prop3.10} and directly computing, we have the following two cases.
	
	$\bf{Case ~1.}$ If $\gamma=0$, then
	$$\begin{aligned}\Omega_4+\Omega_5=
		\begin{array}{c}
			\begin{cases}
				-p^m(p-1) , &\text{if} ~a\neq 0,b=c=0,\\
				&~~\text{or}~ a\in \mathbb{F}_{p^m}, b\neq 0, c \in  \mathbb{F}_p, l_2\neq0,l_3=0,\\
				&~~\text{or}~ a\in \mathbb{F}_{p^m}, b\neq 0, c \in \mathbb{F}_{p},l_2=0,l_3\neq0;\\
				p^m , &\text{if} ~ a\in \mathbb{F}_{p^m}, b\neq 0, c \in \mathbb{F}_{p},l_2 \neq0,l_3\neq0,\\
				&~~\text{or}~a\neq 0, b=0,c\neq0;\\
				p^m(p-1)^2, &\text{if} ~a\in \mathbb{F}_{p^m}, b\neq 0, c \in \mathbb{F}_{p},l_2=l_3=0;\\
				-p^m(p^m-1),&\text{if} ~a=b=0,c\neq 0.
			\end{cases}
		\end{array} .
	\end{aligned}$$
	Furthermore, 
	$$
	wt_2(\textbf{c})=
	\begin{cases}
		p^{2m-2}(p-1) , &\text{if} ~a\neq 0,b=c=0,\\
		&~~\text{or}~ a\in \mathbb{F}_{p^m}, b\neq 0, c \in \mathbb{F}_p, l_2\neq0,l_3=0,\\
		&~~\text{or}~ a\in \mathbb{F}_{p^m}, b\neq 0, c \in \mathbb{F}_{p},l_2=0,l_3\neq0;\\
		p^{m-1}(p^m-p^{m-1}-1) , &\text{if} ~ a\in \mathbb{F}_{p^m}, b\neq 0, c \in \mathbb{F}_{p},l_2\neq0,l_3\neq0,\\
		&~~\text{or}~a\neq 0, b=0,c\neq0;\\
		p^{m-1}(p^m-p^{m-1}-p+1),&\text{if} ~a\in \mathbb{F}_{p^m}, b\neq 0, c \in \mathbb{F}_{p},l_2=l_3=0;\\
		p^{m-1}(p^m-1),&\text{if} ~a=b=0,c\neq 0.
	\end{cases}\label{wt2'}
	$$
	
	$\bf{Case~ 2.}$ If $\gamma\neq0$, then
	$$\begin{aligned}\Omega_4+\Omega_5=
		\begin{array}{c}
			\begin{cases}
				-p^m(p-1) , &\text{if} ~a\neq 0,b=c=0,\\
				&~~\text{or}~ a\in \mathbb{F}_{p^m}, b\neq 0, c \in \mathbb{F}_{p},l_2=l_3=0,\\
				&~~\text{or}~ a\in \mathbb{F}_{p^m}, b\neq 0, c \in \mathbb{F}_{p},l_3\neq0,
				l_2^2-4l_1l_3\in NS_p;\\
				p^m , &\text{if} ~ a\neq 0, b=0,c\neq0,\\
				&~~\text{or}~a\in \mathbb{F}_{p^m}, b\neq 0, c \in \mathbb{F}_{p},l_2\neq0,l_3=0,\\
				&~~\text{or}~a\in \mathbb{F}_{p^m}, b\neq 0, c \in \mathbb{F}_{p},l_3\neq0,
				4l_1l_3-l_2^2=0;\\
				p^m(p+1), &\text{if} ~a\in \mathbb{F}_{p^m}, b\neq 0, c \in \mathbb{F}_p, l_3\neq0,
				l_2^2-4l_1l_3\in S_p;\\
				-p^m(p^m-1),&\text{if} ~a=b=0,c\neq 0.
			\end{cases}
		\end{array} 
	\end{aligned}$$
	Furthermore, 
	$$
	wt_2(\textbf{c})= 
	\begin{cases}
		p^{2m-2}(p-1) , &\text{if} ~a\neq 0,b=c=0,\\
		&~~\text{or}~ a\in \mathbb{F}_{p^m}, b\neq 0, c \in \mathbb{F}_{p},l_2=l_3=0,\\
		&~~\text{or}~ a\in \mathbb{F}_{p^m}, b\neq 0, c \in \mathbb{F}_{p},l_3\neq0,
		l_2^2-4l_1l_3\in NS_p;\\
		
		p^{m-1}(p^m-p^{m-1}-1) ,  &\text{if} ~ a\neq 0, b=0,c\neq0,\\
		&~~\text{or}~a\in \mathbb{F}_{p^m}, b\neq 0, c \in \mathbb{F}_{p},l_2\neq0,l_3=0,\\
		&~~\text{or}~a\in \mathbb{F}_{p^m}, b\neq 0, c \in \mathbb{F}_{p},l_3\neq0,
		4l_1l_3-l_2^2=0;\\
		p^{m-1}(p^m-p^{m-1}-2), &\text{if} ~a\in \mathbb{F}_{p^m}, b\neq 0, c \in \mathbb{F}_p, l_3\neq0,
		l_2^2-4l_1l_3\in S_p;\\
		
		p^{m-1}(p^m-1),&\text{if} ~a=b=0,c\neq 0.
	\end{cases} \label{wt2}
	$$
	
	Now if  $\gamma=0$, then the weight $w_5=p^{2m-2}(p-1),w_6=p^{m-1}(p^m-p^{m-1}-1),w_7=p^{m-1}(p^m-p^{m-1}-p+1)$ and $w_8=p^{m-1}(p^m-1)$,  and so we can get
	$$\begin{aligned}
		&A_{w_5}=(p^{m}-1)(2p^m-2p^{m-1}+1), A_{w_6}=(p-1)(p^m-1)(p^m-p^{m-1}+1),\\
		&A_{w_7}=p^{m-1}(p^{m}-1),A_{w_8}=p-1.
	\end{aligned}$$
	If  $\gamma \neq 0$, then the weight $\tilde{w}_5=p^{2m-2}(p-1),\tilde{w}_6=p^{m-1}(p^m-p^{m-1}-1),\tilde{w}_7=p^{m-1}(p^m-p^{m-1}-2)$ and $\tilde{w}_8=p^{m-1}(p^{m}-1).$
	
	Note that $2 \leq m,$ it's easy to get $p \mid w_i$ and $p \mid \tilde{w}_i (i=5,6,7,8).$ Hence, by Lemma \ref{lemma2.6}, the corresponding code $\overline{\mathcal{C}_{D_2}}$ is  self-orthogonal whether $\gamma=0$ or not.
	
	Finally, we determine the  distance of the dual code ${\overline{\mathcal{C}_{D_2}}}^{\perp}$ for $\gamma = 0.$  By Lemma \ref{lemma2.5} \eqref{(2.1)}-\eqref{(2.4)} and directly computing, we have  $A_1^{\perp}=A_2^{\perp}=0$ and $A_3^{\perp}>0$, which means that  $d({\overline{\mathcal{C}_{D_2}}}^{\perp})=3,$ and by the definition of the projective code,  ${\overline{\mathcal{C}_{D_2}}}$ is a projective linear code.  
	
	From the above, the proof of Theorem \ref{theorem3.4} is complete. \hfill$\square$

	\subsection{Some examples}
	\label{subsection4.2}
	In this subsection, for Theorems \ref{theorem3.1}-\ref{theorem3.4}, we give some examples which are checked by the Magma program.
	\begin{table}[H]
		\centering
		\caption{Some examples for Theorems \ref{theorem3.1}-\ref{theorem3.4}}
		\renewcommand{\arraystretch}{1.5}
		\begin{tabular}{|c|c|c|c|c|}
			\hline
			$p$	& $m$ & Theorem & parameters & weight enumerator \\
			\hline
			\multirow{9}{*}{3} & \multirow{3}{*}{$m=2$} & Theorem \ref{theorem3.1} & $\left[33,5,18 \right]_3 $ & $1+96z^{21}+32z^{18}+112z^{24}+2z^{33}$ \\
			\cline{3-5}
			& & Theorem \ref{theorem3.2} & $\left[ 24,5,12\right]_3$ & $1+104z^{18}+112z^{15}+24z^{12}+2z^{24}$ \\
			\cline{3-5}
			& & Theorem \ref{theorem3.4} & $\left[ 24,5,12\right]_3$ & $1+24z^{12}+112z^{15}+104z^{18}+2z^{24}$ \\
			\cline{2-5}
			& \multirow{3}{*}{$m=3$} & Theorem \ref{theorem3.1} & $\left[261,7,162\right]_3$  & $1+260z^{162}+936z^{171}+988z^{180}+2z^{261}$ \\
			\cline{3-5}
			& & Theorem \ref{theorem3.2}& $\left[234,7,144 \right]_3$  &  $1+962z^{162}+988z^{153}+234z^{144}+2z^{234}$\\
			\cline{3-5}
			& & Theorem \ref{theorem3.4}& $\left[234,7,144 \right]_3$  & $1+234z^{144}+988z^{153}+962z^{162}+2z^{234}$ \\
			\cline{2-5}
			& \multirow{3}{*}{$m=4$} & Theorem \ref{theorem3.1} & $\left[2241,9,1458 \right]_3 $ & $1+2240z^{1458}+8640z^{1485}+8800z^{1512}+2z^{2241}$ \\
			\cline{3-5}
			& & Theorem \ref{theorem3.2} & $\left[2160,9,1404 \right]_3$  & $1+8720z^{1458}+8800z^{1431}+2160z^{1404}+2z^{2160}$  \\
			\cline{3-5}
			& & Theorem \ref{theorem3.4} & $\left[2160,9,1404 \right]_3$  &$1+2160z^{1404}+8800z^{1431}+8720z^{1458}+2z^{2160}$  \\
			\hline
		\end{tabular}
	\end{table}

	\section{Applications for LCD codes }
	\label{section5}
	In this section, we  construct  three classes of   LCD codes  from  self-orthogonal codes given by Section \ref{section3},
	 and prove that  there exists a class of these LCD codes  whose dual codes  are almost optimal  LCD codes according to the sphere packing bound, and give some corresponding examples.
	\begin{theorem}\label{theorem5.1}
		Let $\mathbb{F}_{p^m}^*=\left\langle \omega\right\rangle.$ For the linear code $\overline{\mathcal{C}_{D_1}},$ if  $\gamma=\mathrm{Tr}(-\frac{\beta}{\alpha}),$ then the following two statements are true.
		
		\rm{(1)} $\overline{\mathcal{C}_{D_1}}$ has the generator matrix  
		$$\begin{aligned}
			\boldsymbol{G}_{D_1}
			=&	\begin{bmatrix}
				1 & 1 & \cdots & 1&1 \\
				\operatorname{Tr}(x_1\omega^0) & \operatorname{Tr}(x_2\omega^0)& \cdots &\operatorname{Tr}(x_{n_1-1}\omega^0)& \operatorname{Tr}(x_{n_1}\omega^0) \\
				\operatorname{Tr}(x_1\omega) & \operatorname{Tr}(x_2\omega)& \cdots &\operatorname{Tr}(x_{n_1-1}\omega)& \operatorname{Tr}(x_{n_1}\omega) \\
				\vdots & \vdots &  &\vdots \\
				\operatorname{Tr}(x_1\omega^{m-1}) & \operatorname{Tr}(x_2\omega^{m-1})& \cdots &\operatorname{Tr}(x_{n_1-1}\omega^{m-1})& \operatorname{Tr}(x_{n_1}\omega^{m-1})\\
				\operatorname{Tr}(y_1\omega^{0}) & \operatorname{Tr}(y_2\omega^{0})& \cdots &\operatorname{Tr}(y_{n_1-1}\omega^{0})& \operatorname{Tr}(y_{n_1}\omega^{0})\\
				\vdots & \vdots &  &\vdots&\vdots \\
				\operatorname{Tr}(y_1\omega^{m-1}) & \operatorname{Tr}(y_2\omega^{m-1})& \cdots &\operatorname{Tr}(y_{n_1-1}\omega^{m-1})& \operatorname{Tr}(y_{n_1}\omega^{m-1})\\
			\end{bmatrix},
		\end{aligned}
		$$
		where $D_1=\left\lbrace (x_1,y_1),(x_2,y_2),\cdots,(x_{n_1},y_{n_1})   \right\rbrace \subseteq \mathbb{F}_{p^m}^{2}$  and $n_1=p^{2m-1}+p^m-p^{m-1}$. 
		
		\rm{(2)} The linear code $\mathcal{C}_1$ generated by $\boldsymbol{\mathcal{G}}_1=\left[\boldsymbol{\mathrm{I}}_{2m+1,2m+1}: \boldsymbol{G}_{D_1}\right]$ is LCD with the parameters
		$$\left[p^{m-1}(p^m+p-1)+2m+1,2m+1,\geq p^{2m-2}(p-1)+1 \right]_p.$$
	\end{theorem}
	
	\noindent\textbf{Proof.}
	\textbf{(1)} By $\mathbb{F}_{p^m}^*=\left\langle \omega\right\rangle,$   it's easy to know  that $\{  1,\omega,\omega^2,\cdots,\omega^{m-1}\}   $ is a basis of $\mathbb{F}_{p^m}$ over $\mathbb{F}_{p}$, furthermore, the generator matrix of $\overline{\mathcal{C}_{D_1}}$ can be represented as follows,
	$$\begin{aligned}
		\boldsymbol{G}_{D_1}
		=&	\begin{bmatrix}
			1 & 1 & \cdots & 1&1 \\
			\operatorname{Tr}(x_1\omega^0) & \operatorname{Tr}(x_2\omega^0)& \cdots &\operatorname{Tr}(x_{n_1-1}\omega^0)& \operatorname{Tr}(x_{n_1}\omega^0) \\
			\operatorname{Tr}(x_1\omega) & \operatorname{Tr}(x_2\omega)& \cdots &\operatorname{Tr}(x_{n_1-1}\omega)& \operatorname{Tr}(x_{n_1}\omega) \\
			\vdots & \vdots &  &\vdots \\
			\operatorname{Tr}(x_1\omega^{m-1}) & \operatorname{Tr}(x_2\omega^{m-1})& \cdots &\operatorname{Tr}(x_{n_1-1}\omega^{m-1})& \operatorname{Tr}(x_{n_1}\omega^{m-1})\\
			\operatorname{Tr}(y_1\omega^{0}) & \operatorname{Tr}(y_2\omega^{0})& \cdots &\operatorname{Tr}(y_{n_1-1}\omega^{0})& \operatorname{Tr}(y_{n_1}\omega^{0})\\
			\vdots & \vdots &  &\vdots&\vdots \\
			\operatorname{Tr}(y_1\omega^{m-1}) & \operatorname{Tr}(y_2\omega^{m-1})& \cdots &\operatorname{Tr}(y_{n_1-1}\omega^{m-1})& \operatorname{Tr}(y_{n_1}\omega^{m-1})\\
		\end{bmatrix},
	\end{aligned}
	$$
	where $D_1=\left\lbrace (x_1,y_1),(x_2,y_2),\cdots,(x_{n_1},y_{n_1})   \right\rbrace \subseteq \mathbb{F}_{p^m}^{2} $ and $n_1=p^{2m-1}+p^m-p^{m-1}$. 
	
	\textbf{(2)}  On the one hand, by  Theorem  \ref{theorem3.1},  $\overline{\mathcal{C}_{D_1}}$ is a self-orthogonal code over $\mathbb{F}_p$ with the parameters $[p^{m-1}(p^m+p-1),2m$
	$+1,p^{2m-2}(p-1) ]_p$. Since $\boldsymbol{\mathcal{G}}_1=\left[\boldsymbol{\mathrm{I}}_{2m+1,2m+1}: \boldsymbol{G}_{D_1}\right]$ is the generator matrix of the linear code  $\mathcal{C}_1,$ then by Lemma \ref{lemma2.7}, we know that $\mathcal{C}_1$ is LCD. 
	
	On the other hand,  it's easy to know that the length of $\mathcal{C}_1$ is $p^{m-1}(p^m+p-1)+2m+1$ and the dimension of $\mathcal{C}_1$ is $2m+1.$ Then we only  determine the minimum distance of $\mathcal{C}_1.$ For any non-zero codeword $\boldsymbol{c}  \in \mathcal{C}_1,$ by $\boldsymbol{\mathcal{G}}_1=\left[\boldsymbol{\mathrm{I}}_{2m+1,2m+1}: \boldsymbol{G}_{D_1}\right],$ we know that there exists a non-zero vector $\boldsymbol{u} \in \mathbb{F}_p^{2m+1}$ such that $$\boldsymbol{c}=\boldsymbol{u}\boldsymbol{\mathcal{G}}_1=(\boldsymbol{u},\boldsymbol{u}\boldsymbol{G}_{D_1}),$$ where $\boldsymbol{u}\boldsymbol{G}_{D_1} \in \overline{\mathcal{C}_{D_1}}.$ Then $$wt(\boldsymbol{c})=wt((\boldsymbol{u},\boldsymbol{u}\boldsymbol{G}_{D_1}))=wt(\boldsymbol{u})+wt(\boldsymbol{u}\boldsymbol{G}_{D_1}).$$ 
	Note that $\boldsymbol{u} \neq \boldsymbol{0},$ then 
	$wt(\boldsymbol{u}) \geq 1 ,$  and so we have $$wt(\boldsymbol{c}) \geq wt(\boldsymbol{u}\boldsymbol{G}_{D_1})+1 \geq d(\overline{\mathcal{C}_{D_1}})+1=p^{2m-2}(p-1)+1.  $$
	
	From the above, the proof of Theorem \ref{theorem5.1} is complete. \hfill$\square$

	In the similar proof as that for Theorem \ref{theorem5.1}, one can obtain the following
	
	\begin{theorem}\label{Atheorem5.2}
		
			Let $\mathbb{F}_{p^m}^*=\left\langle \omega\right\rangle.$ For the linear code $\overline{\mathcal{C}_{D_1}},$ if  $\gamma \neq \mathrm{Tr}(-\frac{\beta}{\alpha}),$ then the following two statements are true.

			\rm{(1)} $\overline{\mathcal{C}_{D_1}}$ has the generator matrix  
			$$\begin{aligned}
				\boldsymbol{G}_{D_1}^{(1)}
				=&	\begin{bmatrix}
					1 & 1 & \cdots & 1&1 \\
					\operatorname{Tr}(x_1\omega^0) & \operatorname{Tr}(x_2\omega^0)& \cdots &\operatorname{Tr}(x_{n_1-1}\omega^0)& \operatorname{Tr}(x_{n_1}\omega^0) \\
					\operatorname{Tr}(x_1\omega) & \operatorname{Tr}(x_2\omega)& \cdots &\operatorname{Tr}(x_{n_1-1}\omega)& \operatorname{Tr}(x_{n_1}\omega) \\
					\vdots & \vdots &  &\vdots \\
					\operatorname{Tr}(x_1\omega^{m-1}) & \operatorname{Tr}(x_2\omega^{m-1})& \cdots &\operatorname{Tr}(x_{n_1-1}\omega^{m-1})& \operatorname{Tr}(x_{n_1}\omega^{m-1})\\
					\operatorname{Tr}(y_1\omega^{0}) & \operatorname{Tr}(y_2\omega^{0})& \cdots &\operatorname{Tr}(y_{n_1-1}\omega^{0})& \operatorname{Tr}(y_{n_1}\omega^{0})\\
					\vdots & \vdots &  &\vdots&\vdots \\
					\operatorname{Tr}(y_1\omega^{m-1}) & \operatorname{Tr}(y_2\omega^{m-1})& \cdots &\operatorname{Tr}(y_{n_1-1}\omega^{m-1})& \operatorname{Tr}(y_{n_1}\omega^{m-1})\\
				\end{bmatrix},
			\end{aligned}
			$$
			where $D_1=\left\lbrace (x_1,y_1),(x_2,y_2),\cdots,(x_{n_1},y_{n_1})   \right\rbrace \subseteq \mathbb{F}_{p^m}^{2} $ and $n_1=p^{2m-1}-p^{m-1}$.

			\rm{(2)} The linear code $\mathcal{C}_1^{(1)}$ generated by $\boldsymbol{\mathcal{G}}_1^{(1)}=\left[\boldsymbol{\mathrm{I}}_{2m+1,2m+1}: \boldsymbol{G}_{D_1}^{(1)}\right]$ is LCD with the parameters
			$$\left[p^{m-1}(p^m-1)+2m+1,2m+1,\geq p^{m-1}(p^m-p^{m-1}-2)+1 \right]_p.$$
	\end{theorem}
	\begin{theorem}\label{theorem5.2}
		Let $\mathbb{F}_{p^m}^*=\left\langle \omega\right\rangle.$ For the linear code $\overline{\mathcal{C}_{D_2}},$ if  $\gamma=0,$ then the following two statements are true.
		
		\rm{(1)} $\overline{\mathcal{C}_{D_2}}$ has the generator matrix  
		$$\begin{aligned}
			\boldsymbol{G}_{D_2}
			=&	\begin{bmatrix}
				1 & 1 & \cdots & 1&1 \\
				\operatorname{Tr}(x_1\omega^0) & \operatorname{Tr}(x_2\omega^0)& \cdots &\operatorname{Tr}(x_{n_2-1}\omega^0)& \operatorname{Tr}(x_{n_2}\omega^0) \\
				\operatorname{Tr}(x_1\omega) & \operatorname{Tr}(x_2\omega)& \cdots &\operatorname{Tr}(x_{n_2-1}\omega)& \operatorname{Tr}(x_{n_2}\omega) \\
				\vdots & \vdots &  &\vdots&\vdots \\
				\operatorname{Tr}(x_1\omega^{m-1}) & \operatorname{Tr}(x_2\omega^{m-1})& \cdots &\operatorname{Tr}(x_{n_2-1}\omega^{m-1})& \operatorname{Tr}(x_{n_2}\omega^{m-1})\\
				\operatorname{Tr}(y_1\omega^{0}) & \operatorname{Tr}(y_2\omega^{0})& \cdots &\operatorname{Tr}(y_{n_2-1}\omega^{0})&0\\
				\vdots & \vdots &  &\vdots&\vdots \\
				\operatorname{Tr}(y_1\omega^{m-1}) & \operatorname{Tr}(y_2\omega^{m-1})& \cdots & \operatorname{Tr}(y_{n_2-1}\omega^{m-1})&0\\
			\end{bmatrix},
		\end{aligned}
		$$
		where $D_2=\left\lbrace (x_1,y_1),(x_2,y_2),\cdots,(x_{n_2},0)   \right\rbrace \subseteq \mathbb{F}_{p^m}^{2} $ and $n_2=p^{m-1}(p^m-1).$
		
		\rm{(2)} The linear code $\mathcal{C}_2$ generated by $\boldsymbol{\mathcal{G}}_2=\left[\boldsymbol{\mathrm{I}}_{2m+1,2m+1}: \boldsymbol{G}_{D_2}\right]$ is LCD with the parameters
		$$\left[ p^{m-1}(p^m-1)+2m+1,2m+1, \geq p^{m-1}(p^m-p^{m-1}-p+1)+1\right]_p.$$
	\end{theorem}
	
	
	By taking  $m=2$ and $\alpha=1$  in Theorem \ref{theorem5.2}, the corresponding code $\mathcal{C}_2$ is denoted by $\mathcal{C}_3,$ we prove that there exists an  almost optimal code according to the sphere packing bound.
	
	\begin{theorem}\label{theorem5.3}
		Let $ \mathbb{F}_{p^2}^*=\langle\omega \rangle.$ If $\gamma=0,$ then the following three statements are true.
		
		\rm{(1)} $\mathcal{C}_3$ has the generator matrix $\boldsymbol{\mathcal{G}}_3=\left[ \boldsymbol{I}_5 : \boldsymbol{G}_{D_2}^{(1)}\right] ,$
		where 	$$\begin{aligned}
			\boldsymbol{G}_{D_2}^{(1)}=&	\begin{bmatrix}
				1 & 1 & \cdots & 1 &1 \\
				\operatorname{Tr_1^2}(x_1\omega^0) & \operatorname{Tr_1^2}(x_2\omega^0)& \cdots & \operatorname{Tr_1^2}(x_{n_2-1}\omega^0) &\operatorname{Tr_1^2}(x_{n_2}\omega^0)\\
				\operatorname{Tr_1^2}(x_1\omega)+1 & \operatorname{Tr_1^2}(x_2\omega)+1& \cdots & \operatorname{Tr_1^2}(x_{n_2-1}\omega)+1 &\operatorname{Tr_1^2}(x_{n_2}\omega)+1 \\
				\operatorname{Tr_1^2}(y_1\omega^{0}) & \operatorname{Tr_1^2}(y_2\omega^{0})& \cdots & \operatorname{Tr_1^2}(y_{n_2-1}\omega^{0}) &0\\
				\operatorname{Tr_1^2}(y_1\omega^{1})+1 & \operatorname{Tr_1^2}(y_2\omega^{1})+1&\cdots&\operatorname{Tr_1^2}(y_{n_2-1}\omega^{1})+1&1\\
			\end{bmatrix},
		\end{aligned}
		$$
		$D_2=\left\lbrace (x_1,y_1),(x_2,y_2),\cdots,(x_{n_2-1},y_{n_2-1}),(x_{n_2},y_{n_2})=(x_{n_2},0)\right\rbrace \subseteq \mathbb{F}_{p^2}^{2}, n_2=p(p^2-1) $ and $\mathrm{Tr}_1^2$ is the trace function from $\mathbb{F}_{p^2}$ to $\mathbb{F}_p.$
		
		\rm{(2)} If $\operatorname{Tr_1^2}(\omega)\neq 1$, then $\mathcal{C}_3$ is a $[ p(p^2-1)+5,5, \geq p(p^2-2p+1)+2]_p $  LCD code.
		
		\rm{(3)} $\mathcal{C}_3^{\perp}$ is a $\left[ p(p^2-1)+5,p(p^2-1),3\right]_p $ almost optimal LCD code according to the sphere packing bound and  $\mathcal{C}_3$ is a projective LCD code.
	\end{theorem}

	\noindent\textbf{Proof.} 
	\textbf{(1)}
	By  Theorem \ref{theorem5.2}, we know that $\overline{\mathcal{C}_{D_2}}$ has the generator matrix 
	$$\begin{aligned}
		\boldsymbol{G}_{D_2}=&	\begin{bmatrix}
			1 & 1 & \cdots & 1 &1 \\
			\operatorname{Tr_1^2}(x_1\omega^0) & \operatorname{Tr_1^2}(x_2\omega^0)& \cdots & \operatorname{Tr_1^2}(x_{n_2-1}\omega^0) &\operatorname{Tr_1^2}(x_{n_2}\omega^0)\\
			\operatorname{Tr_1^2}(x_1\omega) & \operatorname{Tr_1^2}(x_2\omega)& \cdots & \operatorname{Tr_1^2}(x_{n_2-1}\omega) &\operatorname{Tr_1^2}(x_{n_2}\omega) \\
			\operatorname{Tr_1^2}(y_1\omega^{0}) & \operatorname{Tr_1^2}(y_2\omega^{0})& \cdots & \operatorname{Tr_1^2}(y_{n_2-1}\omega^{0}) &0\\
			\operatorname{Tr_1^2}(y_1\omega^{1}) & \operatorname{Tr_1^2}(y_2\omega^{1})&\cdots&\operatorname{Tr_1^2}(y_{n_2-1}\omega^{1})&0\\
		\end{bmatrix}.
	\end{aligned}
	$$
	Now we transform the matrix 
	$\boldsymbol{G}_{D_2}$ by elementary row operations. For convenience, we denote the vector $\boldsymbol{a}_u$ be the $u$-th row of the matrix $\boldsymbol{G}_{D_2},$ where $1 \leq u \leq 5,$ then replace  $\boldsymbol{a}_3$ with $\boldsymbol{a}_3+\boldsymbol{a}_1,$ and replace  $\boldsymbol{a}_5$ with $\boldsymbol{a}_5+\boldsymbol{a}_1,$ we obtain the matrix
	$$\begin{aligned}
		\boldsymbol{G}_{D_2}^{(1)}=&	\begin{bmatrix}
			1 & 1 & \cdots & 1 &1 \\
			\operatorname{Tr_1^2}(x_1\omega^0) & \operatorname{Tr_1^2}(x_2\omega^0)& \cdots & \operatorname{Tr_1^2}(x_{n_2-1}\omega^0) &\operatorname{Tr_1^2}(x_{n_2}\omega^0)\\
			\operatorname{Tr_1^2}(x_1\omega)+1 & \operatorname{Tr_1^2}(x_2\omega)+1& \cdots & \operatorname{Tr_1^2}(x_{n_2-1}\omega)+1 &\operatorname{Tr_1^2}(x_{n_2}\omega)+1 \\
			\operatorname{Tr_1^2}(y_1\omega^{0}) & \operatorname{Tr_1^2}(y_2\omega^{0})& \cdots & \operatorname{Tr_1^2}(y_{n_2-1}\omega^{0}) &0\\
			\operatorname{Tr_1^2}(y_1\omega^{1})+1 & \operatorname{Tr_1^2}(y_2\omega^{1})+1&\cdots&\operatorname{Tr_1^2}(y_{n_2-1}\omega^{1})+1&1\\
		\end{bmatrix},
	\end{aligned}
	$$
	where
	$D_2=\left\lbrace (x_1,y_1),(x_2,y_2),\cdots,(x_{n_2-1},y_{n_2-1}),(x_{n_2},y_{n_2})=(x_{n_2},0)\right\rbrace \subseteq \mathbb{F}_{p^2}^{2}$ and $n_2=p(p^2-1).$
	It's easy to know that $\boldsymbol{G}_{D_2}^{(1)}$ is also a generator matrix of $\overline{\mathcal{C}_{D_2}}.$ 
	
	\textbf{(2)}
	By Lemma \ref{lemma2.7} and $\boldsymbol{\mathcal{G}}_3=\left[ \boldsymbol{I}_5 :\boldsymbol{G}_{D_2}^{(1)}\right]$, we know that the linear code $\mathcal{C}_3$ is  LCD.  Obviously, the length of  $\mathcal{C}_3$ is $p(p^2-1)+5$ and the dimension of $\mathcal{C}_3$ is $5.$ 
	
	Next, we determine the minimum distance of  $\mathcal{C}_3.$ Let $d(\overline{\mathcal{C}_{D_2}})$ and $d(\mathcal{C}_3)$ be the minimum distance of $\overline{\mathcal{C}_{D_2}}$ and $\mathcal{C}_3,$ respectively.

	Firstly, let $\tilde{\boldsymbol{c}}$ be a non-zero codeword of $\mathcal{C}_3$, then $\tilde{\boldsymbol{c}}$ can be expressed as $$\tilde{\boldsymbol{c}}=(l,a_0,a_1,b_0,b_1)\mathcal{G}_3=(l,a_0,a_1,b_0,b_1,L_1,L_2,\cdots,L_{n_2}),$$ 
	where $l,a_0,a_1,b_0,b_1\in \mathbb{F}_{p},$ $L_i=\operatorname{Tr_1^2}(ax_i+by_i)+l+a_1+b_1, (x_i,y_i) \in D_2( 1 \leq i \leq n_2), a=a_0+a_1\omega$ and $b=b_0+b_1\omega.$ Note that $$(L_1,L_2,\cdots,L_{n_2})=(\operatorname{Tr_1^2}(ax+by)+l+a_1+b_1)_{(x,y) \in D_2} \in \overline{\mathcal{C}_{D_2}},$$
	we denote $\boldsymbol{c}=(L_1,L_2,\cdots,L_{n_2}).$  Then for $\tilde{\boldsymbol{c}},$ we have
	$$\begin{aligned}wt(\tilde{\boldsymbol{c}})=&wt((l,a_0,a_1,b_0,b_1,L_1,L_2, \cdots,L_{n_2}))\\
		=&wt((l,a_0,a_1,b_0,b_1))+wt((L_1,L_2,\cdots,L_{n_2}))\\
		=&wt((l,a_0,a_1,b_0,b_1))+wt(\textbf{c}).
	\end{aligned}$$  
	Now by $\boldsymbol{c} \in \overline{\mathcal{C}_{D_2}}$ and the proof of Theorem \ref{theorem3.4},
	we have   $wt(\textbf{c}) \in {\left\lbrace w_5,w_6,w_7,w_8\right\rbrace },$ where $$\begin{aligned}
		&w_5=p^{2m-2}(p-1),
		w_6=p^{m-1}(p^m-p^{m-1}-1),\\
		&w_7=p^{m-1}(p^m-p^{m-1}-p+1),
		w_8=p^{m-1}(p^m-1).
	\end{aligned}
	$$
	
	Secondly, depending on  $wt(\textbf{c})=w_5,w_6,w_7$ or $w_8,$  we have the following four cases.


	$\bf{Case ~1.}$ If $wt(\textbf{c})=w_5,$ then
	by the proof of Theorem \ref{theorem3.4}, there exist some
	$l,a_0,a_1,b_0$ and $b_1$  such that  one of the following conditions is true,
	
	$\bf(1.1)$ $a_0+a_1\omega \neq 0$ and $ b_0+b_1\omega=l+a_1+b_1=0;$
	
	$\bf(1.2)$ $b_0+b_1\omega \neq 0, \operatorname{Tr_1^2}(-(b_0+b_1\omega))+l+a_1+b_1\neq0$ and $\operatorname{Tr_1^2}(-(a_0+a_1\omega)(b_0+b_1\omega))=0; $
	
	$\bf(1.3)$ $b_0+b_1\omega \neq 0, \operatorname{Tr_1^2}(-(b_0+b_1\omega))+l+a_1+b_1=0$ and $\operatorname{Tr_1^2}(-(a_0+a_1\omega)(b_0+b_1\omega))\neq 0.$
	
	It's easy to know that  $l=a_1=b_0=b_1=0$ and $a_0 \neq 0,$ thus the condition $\bf(1.1)$ is satisfied. Then   the minimum weight $wt(l,a_0,a_1,b_0,b_1)=1.$   Hence, $$wt(\tilde{\boldsymbol{c}})\geq w_5(\textbf{c})+1=p^{2}(p-1)+1.$$

	$\bf{Case ~2.}$ If $wt(\textbf{c})=w_6,$ then in the similar proof as that for $\bf{Case ~1}$,
	we know that there exist $l=a_0=b_0=b_1=0$ and $a_1 \neq 0$ such that $wt(\textbf{c})=w_6,$ furthermore, $wt(l,a_0,a_1,b_0,b_1)=1.$ Hence,
	$$wt(\tilde{\boldsymbol{c}})\geq w_6(\textbf{c})+1=p(p^2-p-1)+1.$$

	$\bf{Case ~3.}$ If $wt(\textbf{c})=w_7,$ then
	by the proof of Theorem \ref{theorem3.4}, there exist some
	$l,a_0,a_1,b_0$ and $b_1$ such that the following
	three conditions are true simultaneously,
	
	$\bf{(3.1)}$ $b_0+b_1 \omega \neq 0;$
	
	$\bf{(3.2)}$ $\operatorname{Tr_1^2}(-(b_0+b_1\omega))+l+a_1+b_1=0;$
	
	$\bf{(3.3)}$ $\operatorname{Tr_1^2}(-(a_0+a_1\omega)(b_0+b_1\omega))=0. $
	
	Next, to prove  $wt(l,a_0,$ $a_1,b_0,b_1) \geq 2,$ we only prove  the following two statements,
	
	$\bf{1)}$ there do not exist any $l,a_0,a_1,b_0$ and $b_1$ such that $wt(l,a_0,a_1,b_0,b_1)=1;$
	
	$\bf{2)}$ there exist $l,a_0,a_1,b_0$ and $b_1$ such that $wt(l,a_0,a_1,b_0,b_1) = 2.$
	
	{\bf For the statement $1)$} Otherwise, by $wt(l,a_0,a_1,b_0,b_1)=1,$ we know that one of $l,a_0,a_1,b_0$ and $b_1$  is nonzero. Now by  $\bf{(3.1)},$  we have $b_0 \neq 0$ or $ b_1 \neq 0,$ it means that  $l=a_0=a_1=0.$ Then by directly computing,  $\bf{(3.3)}$ holds.

	If $b_0 \neq 0,$ then  $b_1=0.$ By   $\bf{(3.2)}$, we can get
	$$\operatorname{Tr_1^2}(-(b_0+b_1\omega))+l+a_1+b_1=\operatorname{Tr_1^2}(-b_0)=(-b_0)+(-b_0)^p=-2b_0=0,$$ which contradicts the fact  $b_0 \neq 0.$
	
	If $b_1\neq0,$ then  $b_0=0.$ By  $\bf{(3.2)}$, we have
	$$\begin{aligned}
		\operatorname{Tr_1^2}(-(b_0+b_1\omega))+l+a_1+b_1&=\operatorname{Tr_1^2}(-b_1\omega)+b_1\\
		&=-b_1\operatorname{Tr_1^2}(\omega)+b_1\\
		&=b_1(1-\operatorname{Tr_1^2}(\omega))\\
		&=0,
	\end{aligned}
	$$ since $\mathbb{F}_{p^2}^*=\langle\omega \rangle$ and $\operatorname{Tr_1^2}(\omega)\neq 1,$ then we have $b_1=0,$ which contradicts the fact $b_1 \neq 0.$
	
	And so,  there do not exist $l,a_0,a_1,b_0$ and $ b_1$ such that $wt(l,a_0,a_1,b_0,b_1)=1.$

	{\bf For the statement $2)$}  There exist $a_0=a_1=b_1=0, b_0 \neq 0$ and $l=2b_0 \neq 0$ such that  $\bf{(3.1)}$ and $\bf{(3.3)}$ are both true. And by directly computing, we have  
	$$\operatorname{Tr_1^2}(-(b_0+b_1\omega))+l+a_1+b_1=\operatorname{Tr_1^2}(-b_0)+l=(-b_0)+(-b_0)^p+2b_0=-2b_0+2b_0=0,$$
	it means that  $\bf{(3.2)}$ is true.

	In summary of the above discussions,  we know that there exist  $a_0=a_1=b_1=0, b_0 \neq 0$ and $l=2b_0 \neq 0$ such that $wt(\textbf{c})=w_7,$ furthermore, $wt(l,a_0,a_1,b_0,b_1)=2.$
	Hence, $$wt(\tilde{\boldsymbol{c}}) \geq w_7(\textbf{c})+2=p(p^2-2p+1)+2.$$

	$\bf{Case ~4.}$ If $wt(\textbf{c})=w_8,$ then  in the similar proof as that for $\bf{Case ~1},$ we know that there exist $a_0=a_1=b_0=b_1=0$ and $l \neq 0$ such that $wt(\textbf{c})=w_8,$ furthermore, $wt(l,a_0,a_1,b_0,b_1)=1.$ Hence,  $$wt(\tilde{\boldsymbol{c}})\geq w_8(\textbf{c})+1=p(p^2-1)+1.$$

	Finally, from the above four cases, we have $$ \begin{aligned}
		wt(\tilde{\boldsymbol{c}}) \geq &\min\left\lbrace p^{2}(p-1)+1, p(p^2-p-1)+1, p(p^2-2p+1)+2, p(p^2-1)+1 \right\rbrace\\
		&=p(p^2-2p+1)+2.
	\end{aligned}$$ 
	Namely, $d(\mathcal{C}_3) \geq d(\overline{\mathcal{C}_{D_2}})+2= p(p^2-2p+1)+2.$
	Therefore,  $\mathcal{C}_3$ is a $[ p(p^2-1)+5,5, \geq p(p^2-2p+1)+2] $ LCD code.
	
	\textbf{(3)} Note that the dual of an LCD code is also an LCD code, and so  $\mathcal{C}_3^{\perp}$ is also an LCD code with the length $p(p^2-1)+5$ and the dimension $p(p^2-1).$  
	
	Next  we prove the minimum distance of 
	$\mathcal{C}_3^{\perp}$ is 3. In fact, since $\boldsymbol{\mathcal{G}}_3=\left[ \boldsymbol{I}_5 : \boldsymbol{G}_{D_2}^{(1)}\right] $ is the parity-check matrix of $\mathcal{C}_3^{\perp},$ we only prove  the following two statements,
	
	$\bf1)$ there are three columns in $\boldsymbol{\mathcal{G}}_3$ which are  $\mathbb{F}_p$-linearly dependent;

	$\bf 2)$ for  any two columns of $\boldsymbol{\mathcal{G}}_3$ are $\mathbb{F}_p$-linearly independent.
	
	{\bf For the statement 1)} 
	Since $d(\overline{\mathcal{C}_{D_2}}^{\perp})=3$ and $\boldsymbol{G}_{D_2}^{(1)}$ is the parity-check matrix of $\overline{\mathcal{C}_{D_2}}^{\perp},$  there are three columns in $\boldsymbol{G}_{D_2}^{(1)}$ which are  $\mathbb{F}_p$-linearly dependent. By $\boldsymbol{\mathcal{G}}_3=\left[ \boldsymbol{I}_5 :\boldsymbol{G}_{D_2}^{(1)}~\right], $ it's easy to know that there are three columns in $\boldsymbol{\mathcal{G}}_3$ which are  $\mathbb{F}_p$-linearly dependent. This implies that $d(\mathcal{C}_3^{\perp}) \leq d(\overline{\mathcal{C}_{D_2}}^{\perp})=3$.

	{\bf For the statement 2)} Now we divide the following five cases to prove that any two columns of $\boldsymbol{\mathcal{G}}_3$ are $\mathbb{F}_p$-linearly independent. 
	For convenience, we denote $$\boldsymbol{\mathcal{G}}_3=\left[ \boldsymbol{I}_5 :\boldsymbol{G}_{D_2}^{(1)}\right]=\left[\boldsymbol{e}_1,\boldsymbol{e}_2,\cdots,\boldsymbol{e}_5,\boldsymbol{g}_1,\cdots,\boldsymbol{g}_{p(p^2-1)-1},\boldsymbol{g}_{p(p^2-1)} \right].$$
	
	$\bf{Case~1.}$ Since $\boldsymbol{I}_5$ is an identity matrix, then  $\boldsymbol{e}_i$
	and $\boldsymbol{e}_j$ are 	$\mathbb{F}_p$-linearly independent for any $i \neq j$($1 \leq i,j \leq 5$).
	
	$\bf{Case~2.}$ By $d(\overline{\mathcal{C}_{D_2}}^{\perp})=3, \boldsymbol{g}_i$
	and $\boldsymbol{g}_j$ are 	$\mathbb{F}_p$-linearly independent for any $i \neq j$($1 \leq i,j \leq p(p^2-1)$).
	
	$\bf{Case~3.}$ It's easy to know that $\boldsymbol{e}_i$ and $\boldsymbol{g}_{p(p^2-1)}$ are $\mathbb{F}_p$-linearly independent for any $i(1 \leq i \leq 5)$.
	
	$\bf{Case~4.}$ It's easy to know that $\boldsymbol{e}_i$ and $\boldsymbol{g}_{j}$ are $\mathbb{F}_p$-linearly independent for any $i(2 \leq i \leq 5)$ and $j(1 \leq j \leq p(p^2-1)-1)$.
	
	$\bf{Case~5.}$ From $\boldsymbol{e}_1=(1, 0, \ldots, 0)^\top,$
	we know that \( \boldsymbol{e_1} \) and \( \boldsymbol{g_i} \) are \( \mathbb{F}_p \)-linearly dependent for any $i(1 \leq i \leq p(p^2-1)-1)$ if and only if there exist two nonzero elements $a$ and $b$  such that $a\boldsymbol{e}_1+b\boldsymbol{g_i}=0,$ namely, the following system of equations holds,
	\begin{equation}
		\begin{split}
			\begin{cases} 
				x_i + x_i^p = 0 \\ 
				\omega x_i + \omega^p x_i^p + 1 = 0 \\ 
				y_i + y_i^p = 0 \\ 
				\omega y_i + \omega^p y_i^p + 1 = 0
			\end{cases},
			\label{(5.1)}
		\end{split}
	\end{equation}
	where $(x_i,y_i)\in D_2   $ for any $i(1 \leq i \leq p(p^2-1)-1)$. It's easy to know that $x_i=y_i$ if \eqref{(5.1)} holds for any $i(1 \leq i \leq p(p^2-1)-1)$. Since $x_i=y_i$ and $$(x_i,y_i) \in D_2=\left\lbrace (x , y) \in \mathbb{F}_{p^m}^* \times \mathbb{F}_{p^m}  : \mathrm{Tr} ( xy  +x)=0  \right\rbrace,$$ we can deduce that $$\begin{aligned}
		\operatorname{Tr_1^2}(x_iy_i+x_i)&=\operatorname{Tr_1^2}(x_i^2+x_i)=\operatorname{Tr_1^2}(x_i^2)+\operatorname{Tr_1^2}(x_i)\\
		&=x_i^2+x_i^{2p}+x_i+x_i^{p}=x_i^2+x_i^{2p}\\
		&=(x_i+x_i^{p})^2-2x_i^{p+1}\\
		&=-2x_i^{p+1}\neq0,\end{aligned}$$
	which contradicts the fact  $(x_i,y_i)\in D_2 .$ Hence, \( \boldsymbol{e_1} \) and \( \boldsymbol{g_i} \) are \( \mathbb{F}_p \)-linearly independent for any $i(1 \leq i \leq p(p^2-1)-1)$. 
	
	From the above five cases,  any two columns of $\boldsymbol{\mathcal{G}}_3$ are $\mathbb{F}_p$-linearly independent. And then
	according to  $\bf 1)$-$\bf 2)$,  the minimum distance of the dual code $\mathcal{C}_3^{\perp}$ is 3. Now by the definition of the projective code,  $\mathcal{C}_3$ is a projective LCD code.
	
	Moreover, we have $p^{n-(p(p^2-1)+2)}<1+n(p-1),$ where $n=p(p^2-1)+5.$ Then from the sphere packing bound given by Lemma \ref{lemma2.9}, we know that there does not exist any $[ p(p^2-1)+5,p(p^2-1)$ $+2,3]$ code. Therefore, $\mathcal{C}_3^{ \perp}$ is a $\left[ p(p^2-1)+5,p(p^2-1),3\right]_p$ almost optimal LCD code.
	
	From the above, the proof of Theorem \ref{theorem5.3} is complete. \hfill$\square$

	\begin{example}
		By taking  $p=3,m=2$ and $\alpha=\beta=\gamma=1$ in Theorem \ref{theorem3.1}. Let  $ \mathbb{F}_{3^2}^*=\left\langle \omega\right\rangle,$ then $\overline{\mathcal{C}_{D_1}}$ defined by \eqref{(1.4)} has parameters $\left[33,5,18 \right]_{3},$ and by the Magma program, the corresponding defining set 
		$$
		\tilde{D}=\left\lbrace 
		\begin{aligned}
			&(0,\omega), (0,\omega^3), (0,2), (1,\omega),(1,\omega^3),(1,2),(\omega,0),(\omega,1),(\omega,2),
			(\omega^2,\omega^2),(\omega^2,2),(\omega^2,\omega^5),\\&(\omega^3,0),(\omega^3,1),(\omega^3,2),(2,0),(2,1),
			(2,\omega),(2,\omega^2),(2,\omega^3),(2,2),(2,\omega^5),(2,\omega^6),(2,\omega^7),\\
			&(\omega^5,\omega^2),(\omega^5,2),(\omega^5,\omega^5),(\omega^6,2),(\omega^6,\omega^6),(\omega^6,\omega^7),(\omega^7,2),(\omega^7,\omega^6),(\omega^7,\omega^7)
		\end{aligned}
		\right\rbrace.
		$$
		
		By taking  $D_1=\tilde{D}$ in Theorem \ref{theorem5.1} and   directly computing, $\overline{\mathcal{C}_{D_1}}$  has  the generator matrix 
		

		$$\begin{aligned}
			\boldsymbol{G}_{D_1}=&
			\begin{bmatrix}
				\boldsymbol{g}_1 \\
				\boldsymbol{g}_2 \\
				\boldsymbol{g}_3 \\
				\boldsymbol{g}_4\\
				\boldsymbol{g}_5\\
			\end{bmatrix},
		\end{aligned}
		$$
		where $$\begin{aligned}
			\boldsymbol{g}_1&=\left(1,1,1,1,1,1,1,1,1,1,1,1,1,1,1,1,1,1,1,1,1,1,1,1,1,1,1,1,1,1,1,1,1 \right),\\
			\boldsymbol{g}_2&=\left(2,1,1,2,0,1,0,2,2,1,0,2,1,1,1,0,1,0,1,2,1,2,1,2,0,2,0,1,1,1,0,1,0 \right),\\
			\boldsymbol{g}_3&=\left(1,0,2,1,0,2,1,1,2,2,2,0,1,2,2,1,1,2,2,0,0,2,1,2,1,0,0,2,2,0,0,2,2 \right),\\
			\boldsymbol{g}_4&=\left(1,1,1,1,1,1,0,1,2,1,1,2,2,0,2,2,0,2,2,1,2,1,1,0,1,0,1,0,0,0,1,2,0\right),\\
			\boldsymbol{g}_5&=\left(2,2,1,0,0,2,1,1,2,0,2,0,1,0,0,0,0,2,1,2,1,2,2,2,2,1,1,2,1,0,2,2,2\right).
		\end{aligned}$$
	\end{example}
	
	Based on the Magma program, the code  $\mathcal{C}_1$ generated by $\boldsymbol{\mathcal{G}}_1$ is an LCD code with the parameters $\left[ 38,5,20 \right]_3, $ which is consistent with Theorem \ref{theorem5.1}.

	\begin{example}
		By taking  $p=3,m=2,\alpha=\beta=1$ and $\gamma=0$ in Theorem \ref{theorem3.2}. Let  $ \mathbb{F}_{3^2}^*=\left\langle \omega\right\rangle,$ then $\overline{\mathcal{C}_{D_1}}$ defined by \eqref{(1.4)} has parameters $\left[24,5,12 \right]_{3},$ and by the Magma program, the corresponding defining set 
			$$
			\bar{D}=\left\lbrace 
			\begin{aligned}
				&(0,0), (0,\omega^2), (0,\omega^6), (1,1),(1,\omega^5),(1,\omega^7),(\omega,\omega^3),(\omega,\omega^5),(\omega,\omega^6),
				(\omega^2,0),(\omega^2,\omega^3),\\&(\omega^2,\omega^7),(\omega^3,\omega),(\omega^3,\omega^2),(\omega^3,\omega^7),(\omega^5,1),(\omega^5,\omega),
				(\omega^5,\omega^6),(\omega^6,0),(\omega^6,\omega),(\omega^6,\omega^5),\\&(\omega^7,1),(\omega^7,\omega^2),(\omega^7,\omega^3)
			\end{aligned}
			\right\rbrace.
			$$

			By taking  $D_1=\bar{D}$ in Theorem \ref{Atheorem5.2} and   directly computing, $\overline{\mathcal{C}_{D_1}}$ has the generator matrix 
			$$\begin{aligned}
				\boldsymbol{G}_{D_1}^{(1)}=&
				\begin{bmatrix}
					\boldsymbol{f}_1 \\
					\boldsymbol{f}_2 \\
					\boldsymbol{f}_3 \\
					\boldsymbol{f}_4\\
					\boldsymbol{f}_5\\
				\end{bmatrix},
			\end{aligned}
			$$
			where $$\begin{aligned}
				\boldsymbol{f}_1&=\left(1,1,1,1,1,1,1,1,1,1,1,1,1,1,1,1,1,1,1,1,1,1,1,1 \right),\\
				\boldsymbol{f}_2&=\left(2, 1 ,1, 2,0 ,0 ,0 ,2 ,1, 2 ,0 ,2 ,0 ,2, 0, 1 ,2, 0, 1 ,2, 2, 0, 1, 0\right),\\
				\boldsymbol{f}_3&=\left(2, 1, 0 ,1, 1, 0 ,2, 0, 1, 2, 0, 1 ,1 ,0 ,0 ,0 ,2 ,2, 0 ,1, 0 ,2 ,1 ,1\right),\\
				\boldsymbol{f}_4&=\left(0, 2, 2, 2 ,1, 0 ,0, 2 ,0, 2, 0 ,2 ,0, 1, 0 ,1, 1, 1, 0, 2, 0 ,2, 1, 2\right),\\
				\boldsymbol{f}_5&=\left(1, 2 ,0 ,2 ,1 ,0, 0, 1, 1, 1, 2 ,0, 0, 0, 1, 1, 1, 0 ,2, 1, 2, 0 ,0, 2\right).
			\end{aligned}$$

		Based on the Magma program, the code  $\mathcal{C}_1^{(1)}$ generated by $\boldsymbol{\mathcal{G}}_1^{(1)}$ is an LCD code with the parameters $\left[ 29,5,14 \right]_3, $ which is consistent with Theorem \ref{Atheorem5.2}. 
	\end{example}

	\begin{example}
		By taking  $p=3,m=2$ and $\alpha=1,\gamma=0$ in Theorem \ref{theorem3.4}. Let $ \mathbb{F}_{3^2}^*=\left\langle \omega\right\rangle,$ then $\overline{\mathcal{C}_{D_2}}$ defined by \eqref{(1.4)} has parameters $\left[24,5,12 \right]_3,$ and by the Magma program, the corresponding defining set 
		$$
		D^*=\left\lbrace 
		\begin{aligned}
			&(1, 2), (1, \omega^5), (1, \omega^7), (\omega, \omega^3), (\omega, 2), (\omega, \omega^6),(\omega^2, 1), (\omega^2, 2), (\omega^3, \omega), (\omega^3, \omega^2),(\omega^3, 2),\\ & (2, 2),(2, \omega^5), (2, \omega^7), (\omega^5, \omega^3), (\omega^5, 2), (\omega^5, \omega^6), (\omega^6, 1),(\omega^6, 2), (\omega^7, \omega), (\omega^7, \omega^2),\\&(\omega^7, 2), (\omega^2, 0), (\omega^6, 0)
		\end{aligned}
		\right\rbrace.
		$$
		
		By taking  $D_2=D^*$ in Theorem \ref{theorem5.2} and   directly computing, $\overline{\mathcal{C}_{D_2}}$ has the generator matrix 
		
		$$\begin{aligned}
			\boldsymbol{G}_{D_2}=&
			\begin{bmatrix}
				\boldsymbol{h}_1 \\
				\boldsymbol{h}_2 \\
				\boldsymbol{h}_3 \\
				\boldsymbol{h}_4\\
				\boldsymbol{h}_5\\
			\end{bmatrix}
		\end{aligned}
		$$
		where \[
		\begin{aligned}
			\boldsymbol{h_1}&=(1,1,1,1,1,1,1,1,1,1,1,1,1,1,1,1,1,1,1,1,1,1,1,1),\\
			\boldsymbol{h_2}&=(2,2,2,2,2,2,0,0,2,2,2,1,1,1,1,1,1,0,0,1,1,1,0,0),\\
			\boldsymbol{h_3}&=(2,2,2,0,0,0,2,2,1,1,1,1,1,1,0,0,0,1,1,2,2,2,2,1),\\
			\boldsymbol{h_4}&=(1,1,1,2,1,0,2,1,2,0,1,1,1,1,2,1,0,2,1,2,0,1,0,0),\\
			\boldsymbol{h_5}&=(1,0,2,1,1,1,2,1,0,2,1,1,0,2,1,1,1,2,1,0,2,1,0,0).
		\end{aligned}
		\]
	\end{example}
	
	Based on  the Magma program, the code $\mathcal{C}_2$ generated by $\boldsymbol{\mathcal{G}}_2$ is an LCD code with the parameters $\left[ 29,5,14 \right]_3, $ which is consistent with Theorem \ref{theorem5.2}. 
	
	Let $$\boldsymbol{G}_{D_2}^{(1)}=	\begin{bmatrix}
		\boldsymbol{h}_1 \\
		\boldsymbol{h}_2 \\
		\boldsymbol{h}_3+\boldsymbol{h}_1 \\
		\boldsymbol{h}_4\\
		\boldsymbol{h}_5+\boldsymbol{h}_1\\
	\end{bmatrix} $$ in Theorem \ref{theorem5.3}, based on the Magma program, the code $\mathcal{C}_3$ generated by $\boldsymbol{\mathcal{G}}_3=\left[ \boldsymbol{I}_5 :\boldsymbol{G}_{D_2}^{(1)}\right] $ is a projective LCD code with the parameters $\left[ 29,5,14 \right]_3, $ and the dual code  is a $\left[ 29,24,3 \right]_3$ almost optimal LCD code, which is consistent with Theorem \ref{theorem5.3}.
	
	\begin{remark}
		In Table \ref{table4}, we list the known   LCD codes, and compared with the parameters, it's easy to know that the  length,  the  dimension or the distance of LCD codes constructed in this section are different from all known codes listed in Table \ref{table4}. Therefore,  LCD codes constructed in this section are new.
	\end{remark}
	
	\begin{table}[htbp]
		\centering
		\caption{The known infinite families of LCD codes}
		\scriptsize                      
		\renewcommand{\arraystretch}{1.05}   
		\setlength{\tabcolsep}{2pt}      
		\begin{tabularx}{\textwidth}{@{} l l X l @{}}
			\toprule
			$n$ & $[n, k, d]$ & Condition & Ref. \\
			\midrule
			$n = q^l + 1$ & $[n, k, \ge 2(\delta - 1)]_q$ & $3 \le \delta \le \frac{q^{l-1}}{2} + 3$ & \cite{23} Thm 18 \\
			$n = \frac{q^m-1}{q-1}$ & $[n, n - 1 - 2m \lceil \frac{\delta-1}{q-1} \rceil, \ge 2\delta]_q$ & $2 \le \delta \le \frac{q^{m-1}}{2}$ & \cite{23} Thm 32 \\
			$n = \frac{q-1}{N}q^{m-1}+m+1$ & $[n,m+1,\frac{q-1-N}{N}q^{m-1}+2]_q$ & $(q,m) \neq (3,2),\; 1 \le N < q-1,\; N \mid (q-1)$ & \cite{45} Thm 30 \\
			$n = q^{m-1}(q^m+1)+2m+1$ & $[n,2m+1,\ge (q-1)q^{2m-2}+1]_q$ & $q$ is odd $m \ge 2$; $q$ is even $m \ge 3$ & \cite{45} Thm 31 \\
			$n = \frac{q-1}{N}q^{m-1}(q^m+1)+2m+1$ & $[n,2m+1,\ge \frac{q-1}{N}q^{2m-2}(q-1)+1]_q$ & $1 < N,\; N \mid (q-1)$ \newline $q$ is odd $m \ge 2$; $q$ is even and $m \ge 3$ & \cite{45} Thm 32 \\
			$n = p^m + 2m + 1$ & $[n, 2m + 1, p^m - p^{m-1} - p^{\frac{m-1}{2}} + 1]_p$ & $p \equiv 1 \pmod{4}, m$ is odd \newline $p \equiv 3 \pmod{4}, m \equiv 3 \pmod{4}$ & \cite{29} Thm 40,48 \\
			$n = p^m + 2m + 1$ & $[n, 2m + 1, p^m - p^{m-1} - p^{\frac{m-1}{2}} + 2]_p$ & $p \equiv 3 \pmod{4},\; m \equiv 1 \pmod{4}$ & \cite{29} Thm 40,48 \\
			$n = p^m + 2m + 1$ & $[n, 2m + 1, p^m - p^{m-1} - (p-1)p^{\frac{m-2}{2}} + 1]_p$ & $p \equiv 1 \pmod{4}, m$ is even \newline $p \equiv 3 \pmod{4}, m \ge 4$ is even & \cite{29} Thm 40,48 \\
			$n = p^m + 2m + 1$ & $[n, 2m + 1, p^m - p^{m-1} - (p-1)p^{\frac{m-2}{2}} + 2]_p$ & $p \equiv 3 \pmod{4},\; m = 2$ & \cite{29} Thm 40,48 \\
			$n = p^m + 2m + 1$ & $[n, p^m, 3]_p$ & $(p, m) \neq (3, 2)$ & \cite{26} Thm 40,48 \\
			$n = p^m + m + 2$ & $[n, m+2, d']_p$ & $p$ is odd and $m \in \mathbb{Z}^{+}$ & \cite{38} Prop 3 \\
			\bottomrule
			\label{table4}
		\end{tabularx}
	\end{table}

	\section{Conclusions}
	\label{section6}
	In this manuscript, three classes of  four-weight self-orthogonal projective codes  are constructed. And then, we obtain three classes of  LCD codes. In particular, we prove that there exists a class of  these LCD codes whose dual codes are almost optimal LCD codes according to the sphere packing bound. Comparing  with the parameters of  LCD codes  in Table \ref{table4}, LCD codes  constructed in this manuscript are  new.

	\section*{Acknowledgments}
	This paper is supported by National Natural Science Foundation of China (Grant No. 12471494) and Natural Science Foundation of Sichuan Province (2024NSFSC2051).
	
	\section*{Declarations}
	{\bf Conflict of interest}  The authors declare there are no conflicts of interest.

\end{document}